\documentclass[english,12pt]{article}
\usepackage[letterpaper,top=1.9cm, bottom=1.9cm, left=1.75cm, right=1.75cm]{geometry}
\usepackage[T1]{fontenc}
\usepackage[latin9]{inputenc}
\usepackage{babel}
\usepackage{amsmath}
\usepackage{amsthm}
\usepackage{amsfonts}
\usepackage{layout}
\usepackage{latexsym}
\usepackage[retainorgcmds]{IEEEtrantools}
\usepackage{graphicx}
\usepackage{color}
\usepackage{dsfont}
\usepackage{appendix}
\usepackage{hyperref}
\usepackage{bbm}
\usepackage{bm}
\usepackage{exscale}
\usepackage[normalem]{ulem}
\usepackage{mathrsfs}
 \usepackage{enumerate}

\usepackage{float}

\newcommand{\ud}{\,\mathrm{d}}
\newcommand{\R}{\mathbb{R}}

\definecolor{DarkGreen}{rgb}{0.2,0.6,0.2}

\def\red#1{\textcolor{red}{#1}}

\def\eps{\varepsilon}

\def\CBV{\text{\sl CBV}}

\numberwithin{equation}{section}

\def\Ind#1{{\mathbbmss 1}_{_{\scriptstyle #1}}}

\def\ua{\uparrow}

\def\wh{\widehat}
\def\wt{\widetilde}
\def\bbar{\overline}

\def\ignore#1{}

\def\bR{{\mathbb R}}

\def\bN{\mathbb N}

\def\bT{\mathbb T}

\title{\textbf{Model-free portfolio theory\\ and its functional master formula}}
\author{ \normalsize Alexander Schied\footnote{Department of Statistics and Actuarial Science, University of Waterloo, {\tt aschied@uwaterloo.ca}}\qquad  Leo Speiser\setcounter{footnote}{3}\footnote{Department of Mathematics, University of Mannheim, {\tt speiser.leo@gmail.com}} \qquad  Iryna Voloshchenko\setcounter{footnote}{6}\footnote{
Department of Mathematics, University of Mannheim,  {\tt irynaice@gmail.com} \hfill\break A.S.~gratefully acknowledges partial  supported by Deutsche Forschungsgemeinschaft  DFG through RTG 1953 and  by  the
 Natural Sciences and Engineering Research Council of Canada through grant RGPIN-2017-0405.\hfill\break
 I.V.~gratefully acknowledges support by Deutsche Forschungsgemeinschaft DFG through the Research Training Group RTG 1953.}}

\date{\normalsize First version: June 10, 2016\\
\normalsize This version: May 23, 2018}

\begin{document}	

\maketitle 

\begin{abstract} 
We use pathwise It\^o calculus  to prove  two strictly pathwise versions of the master formula in Fernholz' stochastic portfolio theory.  Our first version is set within the framework of F\"ollmer's pathwise It\^o calculus and works for portfolios generated from functions that may depend on the current states of the market portfolio and an additional path  of finite variation. The second version is formulated within the functional pathwise It\^o calculus of Dupire (2009) and Cont \& Fourni\'e (2010) and allows for portfolio-generating functionals that may depend additionally on the entire path of the market portfolio. Our results are illustrated by several examples and shown to work on empirical market data. 
\end{abstract}


\newtheorem{thm}{Theorem}[section]
\newtheorem{lemma}[thm]{Lemma}
\newtheorem{prop}[thm]{Proposition}
\newtheorem{cor}[thm]{Corollary}
 \theoremstyle{definition}
\newtheorem{definition}[thm]{Definition}
\newtheorem{rem}[thm]{Remark}
\newtheorem{ex}[thm]{Example}
\newtheorem{pr}[thm]{Problem}
\newtheorem{assump}[thm]{Assumption}

\theoremstyle{plain}

\section{Introduction} 

The purpose of this paper is twofold. On the one hand, it deals with an extension of the master formula in stochastic portfolio theory to path-dependent portfolio generating functions. On the other hand, it yields a new case study in which continuous-time trading strategies can be constructed in a probability-free manner by means of pathwise It\^o calculus.

Stochastic portfolio theory (SPT) was  introduced  by Fernholz~\cite{F_div, F_pfgenfct, F};  see also  Karatzas and Fernholz~\cite{FK08} for an overview. Its goal is to construct investment strategies that outperform a certain reference portfolio such as the market portfolio $\mu(t)$; see, e.g.,~\cite{FKK05,VK15,KR16,Ruf18}. Here, our focus is mainly on functionally generated portfolios, which in standard SPT are generated from functions $G(t,\mu(t))$ of the current state, $\mu(t)$, of the market portfolio. The performance of the functionally generated portfolio relative to the market portfolio can be described in a very convenient way by the so-called \emph{master formula} of SPT. See Strong~\cite{Strong2014} for an extension of the master formula to the case in which $G$ may additionally depend on the current state of  a continuous trajectory of bounded variation. 

The first contribution of this paper concerns the basis for the modeling framework of SPT. While price processes for SPT are usually modeled as It\^o processes, it has often been remarked that both sides of the  master formula can be understood in a strictly pathwise manner. So the question arises to what extent a stochastic model is actually needed in setting up SPT. Do price processes really need to be modeled as It\^o processes driven by Brownian motion, or is it possible to relax this condition and consider more general processes, perhaps even beyond the class of semimartingales? Can one get rid of the nullsets that are inherent in stochastic models? That is,  can one prove the master formula strictly path by path?

Our approach gives   affirmative answers  to the questions raised in the preceding paragraph. To this end, we show that SPT can be formulated within the pathwise It\^o calculus introduced by F\"ollmer~\cite{Ito_F} and  further developed to path-dependent functionals by Dupire~\cite{Dupire} and Cont and Fourni\'e~\cite{CF, CF13}. Thus, the only assumption on the trajectories of the price evolution is that they are continuous and admit  quadratic variations and covariations in the sense of~\cite{Ito_F}. This assumption is  satisfied by all typical sample paths of a continuous semimartingale but also by non-semimartingales, such as fractional Brownian motion with Hurst index $H\ge1/2$ and many deterministic fractal curves~\cite{Mishura,Schied_Takagi}.
In this sense, our paper is also a contribution to \emph{robust finance}, which aims to reduce the reliance on a probabilistic model and, thus,  to model uncertainty; see, e.g.,~\cite{bick,cox,Davis,Schied07,Schied13,SchiedVoloshchenko} for similar analyses on other financial problems. Robustness results for discrete-time SPT were previously obtained  
also by Pal and Wong~\cite{Pal16}, where the relative performance of portfolios with respect to a certain benchmark 
is analyzed using the discrete-time energy-entropy framework~\cite{Palforthcoming,Wong15a,Wong15b}. 

To discuss the second contribution of this paper, note that 
in practice portfolios are often constructed not just from current market prices or capitalizations but also from past data, such as econometric estimates,  moving or rolling averages, running maxima, realized covariance, Bollinger bands, etc. It is therefore a natural question whether it is possible to develop a  master formula for portfolios that are generated by functionals of the entire past evolution, $\mu^t:=(\mu(s\wedge t))_{s\ge0}$, of the market portfolio and maybe also other factors. In this paper, we give an affirmative answer to this question. Our main result, Theorem~\ref{Path-dependent pathwise master formula}, contains a master formula for portfolios that are generated by sufficiently smooth functionals of the form $G(t,\mu^t,A^t)$, where $A$ is an additional $m$-dimensional continuous trajectory of bounded variation that may depend on $\mu$ in an adaptive manner.  We then turn to analyzing several concrete examples for portfolios that are generated by functions of mixtures of current market portfolio weights and their moving averages. Our analysis is carried out both on a mathematical level and with empirical market data.

The paper is organized as follows. In Section~\ref{non-pathdep setting section}, we first provide a master formula based on F\"ollmer's \cite{Ito_F} pathwise It\^o calculus. It works for  portfolios generated by functions that, as in Strong \cite{Strong2014}, may depend on the current states of the market portfolio and an additional continuous trajectory $A$ of bounded variation. In this case, both the formulations and the proofs of many results from SPT, including the master formula, can be extended relatively easily to the pathwise setting. In Section~\ref{path-dependent master formula section}, we extend the results from Section~\ref{non-pathdep setting section} to portfolios that may depend on the entire past evolutions of the market portfolio and $A$. To this end, we use the pathwise functional It\^o calculus developed by Dupire~\cite{Dupire} and Cont and Fourni\'e~\cite{CF}. The main difficulty in achieving this extension is that now the It\^o integral is based on \lq\lq Riemann sums'' involving approximations of the integrator path. Therefore, a functional dependence on the integrator paths must be retained in the integrands, and   new arguments are needed so as to prove, e.g., the corresponding master formula. Our above-mentioned examples are discussed in Section~\ref{Examples section}. All proofs are given in Section~\ref{Proofs section}, and key concepts from functional It\^o calculus are recalled in the Appendix.

\goodbreak

\section{Statement of main results} \label{Sec2}

Throughout this paper, we work in a strictly pathwise setting. Our goal is to derive, first, a pathwise master formula for Stochastic Portfolio Theory (SPT)~\cite{F_div, F_pfgenfct, F, FK08}. To this end, we will use the pathwise It\^o calculus developed by F\"ollmer~\cite{Ito_F}. In a second step, we will use the functional extension of pathwise It\^o calculus, as developed by Dupire~\cite{Dupire} and Cont and Fourni\'e~\cite{CF}, so as to extend the pathwise master formula also to path-dependent functionals.

In the sequel, we consider a financial market model  consisting of $d$ risky assets and a locally riskless bond. 
The price
of the bond 
is given by 
$$ B(t) = \exp\Big(\int_0^tr(s)\,\ud s\Big),
$$
where  $r:[0,\infty) \rightarrow \R$ is a measurable short rate function  satisfying 
 $\int_{0}^{T} \vert r(s) \vert \, \mathrm{d}s < \infty$ for all $T>0$. The prices of the risky assets are described by a single $d$-dimensional continuous  path $S:[0,\infty)\to\mathbb R^d$. We emphasize that no probabilistic assumptions are made on the dynamics of $r$ and $S$. All that we require is that the components of $S$, besides being strictly positive, admit continuous covariations  in the pathwise sense proposed by F\"ollmer~\cite{Ito_F}. To recall this notion, let $\left(\mathbb{T}_n\right)_{n\in\mathbb N}$ be a refining sequence of partitions  of $[0,\infty)$, which will remain fixed for the remainder of this paper.
 That is, for fixed $n$, the partition  $\bT_n=\{t_0,t_1,\dots\}$ is such that $0= t_0<t_1<\dots$ and $t_k\to\infty$ as $k\to\infty$. Moreover, we have $\bT_1\subset\bT_2\subset\cdots,$ and the mesh of $\bT_n$ tends to zero on each compact interval as $n\ua\infty$. For fixed $n$, it will be convenient to denote the successor of $t\in\bT_n$ by $t'$, i.e., 
$$t'=\min\{u\in\bT_n\,|\,u>t\}.
$$
We then assume that for   $1\leq i,j\leq d$ and  $t\geq 0$ the sequence
\begin{equation}\label{xin}
 \sum_{s\in \mathbb{T}_n\atop s\le t}  \left(S_i(s')-S_i(s)\right)\left(S_j(s')-S_j(s)\right)
\end{equation}
converges to a finite limit, called the pathwise covariation of $S_i$ and $S_j$ and denoted by $[ S_i,S_j] (t)$, such that $t\to [ S_i,S_j] (t)$ is continuous. As usual, we write $[S_i]:=[S_i,S_i]$ and call this the pathwise quadratic variation of the real-valued path $S_i$. The class of all  continuous functions $S:[0,\infty)\to\R^d$ satisfying the preceding requirement will be denoted by $QV^d$. As mentioned above, we will require in addition that $S_i(t)>0$ for all $i$ and $t$. The corresponding subset of $QV^d$ will be denoted by $QV_+^d$.

Note that $QV^d$ depends strongly on the choice of the partition sequence $(\mathbb{T}_n)_{n\in\mathbb N}$ and is typically not a vector space~\cite{Schied_Takagi}. Moreover, polarization of the sums in \eqref{xin} implies that $[ S_i,S_j] $ exists for $S_i,S_j\in QV^1$, if and only if the quadratic variation $[S_i+S_j]$ exists and that
\begin{equation}\label{polarization identity}
[S_i,S_j](t)=\frac12\Big([S_i+S_j](t)-[S_i](t)-[S_j](t)\Big).
\end{equation}
Thus, for $d>1$, the assumption that the covariations $[S_i,S_j]$ exist cannot be reduced to the existence of the quadratic variations $[S_i]=[S_i,S_i]$.  If $S$ is a typical sample path of a continuous semimartingale then $[S_i,S_j]$ clearly coincides with the quadratic variation taken in the usual sense---provided that this sample path does not belong to a certain nullset, which in turn depends strongly on the partition sequence $(\mathbb{T}_n)_{n\in\mathbb N}$. As a matter of fact, the union of all these nullsets is equal to the entire sample space $\Omega$, because for every continuous function there exists some refining sequence of partitions along which this function has vanishing quadratic variation; see \cite[p.~47]{Freedman}.   Some authors, such as \cite{Ananova} or \cite{CF}, make the dependence of the pathwise quadratic variation (and the pathwise It\^o integral) on the sequence of partitions $\left(\mathbb{T}_n\right)_{n\in\mathbb N}$ explicit by including a corresponding symbol in the notation.   We refrain from this so as not to burden the notation.

According to~\cite{Ito_F},  the requirement $S\in QV^d$ guarantees that It\^o's formula   holds in a pathwise sense.  By taking those functions that appear naturally inside the It\^o integral from It\^o's formula, we arrive at a natural class of  admissible integrands for It\^o integrals with integrator $S$. In Section~\ref{non-pathdep setting section}, we will use the corresponding integration theory to develop a strictly pathwise theory of SPT and state the corresponding master formula. In this context, it will be possible to extend the arguments used, e.g., in~\cite{FK08}. In Section~\ref{path-dependent master formula section}, we will then consider a further extension to path-dependent portfolios that appear in the functional extension of It\^o's formula given recently by Dupire~\cite{Dupire} and Cont and Fourni\'e~\cite{CF}. Here, additional care will be needed in setting up the problem formulation, and a new proof strategy will be needed for the corresponding master formula.

\subsection{The master formula within F\"ollmer's pathwise It\^o calculus}\label{non-pathdep setting section}

We refer to~\cite{Sondermann} and \cite[Section 3]{Schied13} for background on Föllmer's pathwise It\^o calculus, including an English translation of~\cite{Ito_F} as provided in the Appendix of~\cite{Sondermann}. For open sets $U\subset\mathbb R^d$ and $V\subset \mathbb R^n$, the class  $C^{2,1}(U,V)$ will consist of all functions $f:U\times V\to\mathbb R$ that are continuously differentiable in $( x,a)\in U\times V$ and twice continuously differentiable in $x\in U$.  We will write $f_{x_k}=\partial f/\partial x_k$ and  $f_{a_k}=\partial f/\partial a_k$ for the partial derivatives of $f$ with respect to the $k^{\text{th}}$ components of $x=(x_1,\dots, x_d)$ and  $ a=(a_1,\dots, a_n)$, respectively. The gradient of $f$ in direction $ x=(x_1,\dots, x_d)$ will be denoted by 
$\nabla_{  x}f=\big(f_{ x_1},\dots, f_{x_d}\big)$,
 and we will write $f_{x_k,x_m}$ for the  second partial derivatives with respect to the components $x_k$ and $x_m$ of the vector $ x$. The Euclidean inner product of  two vectors $  v$ and $ w$ will be denoted by $ v\cdot w$. We let $\mathbb R^d_+$ be set of all those vectors in $\mathbb R^d$ that have only strictly positive components. The space $\CBV([0,T],V)$ will consist of all continuous  functions $A:[0,T]\to V$ whose components are of bounded variation.  
 The following definition is taken from \cite[Definition 11]{Schied13}.

\begin{definition}
A function $ \xi:[0,\infty)\to\bR^d$ is called an \emph{admissible integrand for $S$} if for each $T>0$ there exists $n\in\bN$,  open sets $U\subset\mathbb R^d$ and   $V\subset\mathbb R^n$,  a function $f\in C^{2,1}(U,V)$, and  $ A\in \CBV([0,T],V)$ such that $S(t)\in U$ and  $\xi(t)=\nabla_{  x}f(S(t),A(t))$ for $0\le t\le T$.
\end{definition}

If $\xi$ is an admissible integrand for $S$, then Föllmer's pathwise It\^o formula, e.g., in the form of \cite[Theorem 9]{Schied13}, implies that for every $T>0$ the  pathwise It\^o integral exists as the following limit of Riemann sums:
\begin{equation}\label{Foellmer integral Riemann sums}
\int_0^T\xi(t)\ud S(t)=\lim_{n\uparrow\infty}\sum_{t\in\mathbb T_n,\,t\le T}\xi(t)\cdot(S(t')-S(t)).
\end{equation}
This integral will be called the \emph{F\"ollmer integral} in the sequel.
Suppose that $\xi$ is an admissible integrand for $S$ and $\eta$ is a real-valued measurable function on $[0,\infty)$ such that $\int_{0}^{t} \vert \eta(s) \vert \, \mathrm{d} \vert B \vert (s) < \infty $ for all $t>0$, where $\ud |B|(s)$ denotes Stieltjes integration with respect to the total variation of $B$. Then the pair $(\xi, \eta )$ is called a \emph{trading strategy}. As usual, the interpretation is that  $\xi_{i}(t)$ corresponds to the number of shares held  in the $i^{\text{th}}$ stock at time $t$ and  $\eta(t)$ is the number of shares held in the riskless bond at time $t$.

\begin{definition} 
Let $\xi$ be an admissible integrand for $S$ and $\eta$ a real-valued measurable function on $[0,\infty)$ such that $\int_{0}^{T} \vert \eta(s) \vert \, \mathrm{d} \vert B \vert (s) < \infty $ for all $T>0$.  The trading strategy $(\xi, \eta )$ is called \emph{self-financing} if the associated wealth $
V(t) := \xi(t) \cdot S(t) + \eta(t) B(t) 
$
satisfies the identity
\begin{equation}
V(t) = V(0) + \int_0^t \xi(s) \, \mathrm{d} S(s) + \int_0^t \eta(s) \,\mathrm{d}B(s),\quad t\geq 0.
\label{Vsf}\end{equation}
\end{definition}

Note that the first integral in \eqref{Vsf} is a Föllmer integral, whereas $\int_0^t \eta(s) \,\mathrm{d}B(s)$ can be understood in the Riemann--Stieltjes sense, since both $\eta(t)$ and $B(t)$ are continuous functions of $t$. Indeed, $V(t)$ is continuous by the continuity of $t\mapsto \int_0^t \xi(s) \, \mathrm{d} S(s)$ and  \eqref{Vsf}, the continuity of $\xi(t)$ follows from the definition of an admissible integrand, and therefore $\eta(t)=(V(t)-\xi(t)\cdot S(t))/B(t)$ is continuous as well.

\begin{lemma}\label{log lemma}
 Suppose that $X$ is a continuous function from $[0,\infty)$ to $\R_+^d$. 
Then $X\in QV^d_+$ if and only if $\log X:=(\log X_1,\dots, \log X_d)^\top\in QV^d$. In this case,
\begin{equation}\label{log covariation id}
[\log X_i,\log X_j](t)=\int_0^t\frac1{X_i(s)X_j(s)}\ud [X_i,X_j](s).
\end{equation}
\end{lemma}

The preceding lemma implies in particular that $\log S=(\log S_1,\dots,\log S_d)^\top\in QV^d$. 
The following  lemmas are standard in stochastic calculus. In  our pathwise setting, however, their proofs need some additional care.

\begin{lemma}\label{pi lemma}
A function $\pi:[0,\infty)\to\R^d$ is an admissible integrand for $\log S$ if and only if $\tfrac{\pi(t)}{S(t)}:= \left(  \tfrac{\pi_{1}(t)}{S_{1}(t)},\ldots,\tfrac{\pi_{d}(t)}{S_{d}(t)}  \right)^\top$ is an admissible integrand for $S$. In this case, 
\begin{equation}\label{log S change eq}
\int_0^t\pi(s)\ud \log S(s)=\int_0^t\frac{\pi(s)}{S(s)}\ud S(s)-\frac12\sum_{i=1}^d\int_0^t\frac{\pi_i(s)}{S^2_i(s)}\ud [S_i](s).
\end{equation}
\end{lemma}

The preceding lemma gives rise to a special class of self-financing trading strategies:

\begin{lemma}\label{portfolio value lemma}Suppose that $\pi$ is an admissible integrand for $\log S$. If we let
\begin{equation}
 \label{ItoDiffEqLSGcadlag}
\begin{split} \lefteqn{V^{\pi} (t):=}\\
&\exp \left(   \int_0^t \dfrac{\pi(s)}{S(s)} \,\mathrm{d} S(s) - \dfrac{1}{2} \sum\limits_{i,j = 1}^{d} \int_0^t \dfrac{\pi_{i}(s) \pi_{j}(s)}{S_{i}(s)S_{j}(s)} \,\mathrm{d} [ S_{i}, S_{j} ] (s)+\int_0^t \Big(  1- \sum\limits_{i=1}^{d} \pi_{i}(s)  \Big) r(s) \,\mathrm{d}s \right),
 \end{split}
\end{equation} 
then 
\begin{align}\label{handelsstrXIETA}
 \xi_{i}(t) := \dfrac{\pi_{i}(t)V^{\pi}(t)}{S_{i}(t)}, \quad i = 1,\ldots,d,\quad \text{and}\quad\eta(t) := \dfrac{\big( 1 - \sum_{i =1}^{d} \pi_{i}(t) \big) V^{\pi} (t) }{B(t)}
  \end{align} 
defines a self-financing strategy with wealth $V^\pi$. \end{lemma}

The preceding lemma justifies the following definition.

\begin{definition}\label{portfolio}
 A function $\pi:[0,\infty)\to\mathbb R^d$ is called a \emph{portfolio for $S$} if it is an admissible integrand for $\log S$   and if 
\begin{equation}\label{portfolioeq}\pi_{1}(t) + \cdots + \pi_{d}(t) = 1,\quad t\geq 0.
\end{equation} 
 A portfolio for $S$ is called \emph{long-only} if $\pi_i(t)\geq 0$ for all $t$. The function $V^\pi$ from \eqref{ItoDiffEqLSGcadlag} will be called the \emph{portfolio value} of $\pi$ with unit initial wealth.
 \end{definition}

As in \cite[Section 2]{FK08},  we normalize the market, i.e., we suppose that at any time $t$ each stock has only one share outstanding. Then the stock prices $S_{i}(t)$ represent the capitalizations of the individual companies.

\begin{lemma}\label{market portfolio lemma} The quantities
$$\mu_{i}(t) := \dfrac{S_{i}(t)}{S_{1}(t) + \cdots + S_{d}(t)}, \quad i = 1,\ldots,d,
$$
form a portfolio for $S$, called the \emph{market portfolio}. The corresponding portfolio value with unit initial wealth is given by
\begin{equation}\label{market portfolio value eq}
V^\mu(t)=\frac{S_1(t)+\cdots +S_d(t)}{S_1(0)+\cdots +S_d(0)}.
\end{equation}
\end{lemma}

As in~\cite{FK08} and~\cite{Strong2014}, we can now consider 
 the pathwise dynamics of a portfolio that is generated by a \emph{portfolio-generating function}. 
 These are $C^{2,1}$-functions that may depend on the current composition $\mu(t)$ of the market  portfolio and on certain other paths $A_0(t),\dots,A_m(t)$ of  (locally) finite variation.   Examples for these additional paths include time, $A_0(t)=t$, running maximum, $A_j(t)=\max_{s\le t}S_i(s)$, moving average, $A_k(t)=\frac1\theta\int_{t-\theta}^tS_i(0\vee s)\,ds$, or realized quadratic variation, $A_\ell(t)=[S_i](t)$; see also Section \ref{Examples section}. We denote by
$\Delta^{d}$
the standard simplex in $\bR^d$ and let $ \Delta^{d}_{+} := \Delta^d\cap\mathbb R^d_+$.

\begin{lemma}\label{G generated portfolio lemma}Let $V\supset  \Delta_{+}^{d}$  and $W\subset \mathbb{R}^ m$  be open sets,  $G$ a strictly positive function in $C^{2,1}(U,W)$, and $A:[0,\infty)\to W$ a continuous function whose restriction to compact intervals is of finite variation. Then 
 \begin{equation} \label{portfolioGeneratedByG}
\pi_{i}(t) =\mu_i(t)+ \mu_i (t)\frac{\partial}{\partial {x_{i}}} \log G(\mu(t),A(t))  - \sum_{j = 1}^{d} \mu_i(t)\mu_{j}(t) \frac{\partial}{\partial {x_{j}}}\log G(\mu(t), A(t)) , \quad 1\leq i \leq d,
\end{equation} 
is a portfolio for $S$ and called the \emph{portfolio generated by} $G$. Moreover, $\pi$ is an admissible integrand for $\log\mu$, where $\log\mu(t):=(\log\mu_1(t),\dots,\log\mu_d(t))$ is defined in analogy to $\log S$.
\end{lemma}

 The following theorem extends the \lq\lq master equation\rq\rq\ from Fernholz~\cite{F_pfgenfct} and Strong~\cite{Strong2014} to our strictly pathwise setting. 

\begin{thm}[\textbf{Pathwise master equation}] \label{MasterFormulaCadlag}
For $G$ as in Lemma~\ref{G generated portfolio lemma}, the relative wealth  of the portfolio $\pi$ generated by $G$  with respect to the market portfolio is given by the following  master equation
\begin{equation*} 
\log \left( \dfrac{V^{\pi}(T)}{V^{\mu}(T)} \right) = \log \left( \dfrac{G( \mu(t),A(T))}{G( \mu(0),A(0))} \right) + \mathfrak{g}([0,T])+ \mathfrak{h}([0,T])  , \quad 0 \leq T < \infty,
\end{equation*}
where  the (possibly signed) Radon measures $\mathfrak g$ and $\mathfrak h$ are given by
\begin{equation*}
\mathfrak{g}(\mathrm{d}t) := - \frac12\sum_{i,j=1}^{d}  \frac{G_{x_i,x_j}( \mu(t),A(t))}{ G( \mu(t),A(t))} \ud [\mu_i,\mu_j](t)\quad\text{and}\quad \mathfrak{h}(\mathrm{d}t) := -\sum_{k=1}^m \dfrac{G_{a_k}(\mu(t),A(t))}{G(\mu(t),A(t))}    \ud A_k(t).
\end{equation*}
\end{thm}

\subsection{The master formula for path-dependent portfolios}\label{path-dependent master formula section}

Our goal in this section is to extend the pathwise portfolio theory from Section~\ref{non-pathdep setting section} to portfolios that are path-dependent and to prove a corresponding master formula. The main difficulty we are encountering here is that the pathwise It\^o integral is then no longer the limit of ordinary Riemann sums as in 
\eqref{Foellmer integral Riemann sums}. Instead, the integrands in the approximating \lq\lq Riemann sums'' \eqref{Itointcont} involve approximations of the integrator path. Therefore, one either needs to make strong regularity assumptions as in  \cite[Theorem 3.2]{Ananova} so as to be able to replace  \eqref{Itointcont} by standard Riemann sums, or to retain the functional dependence in the integrands. Here, we take the latter route. It will require a proof strategy that is different from the one we use to prove Theorem~\ref{MasterFormulaCadlag}.

Let $S\in QV^d_+$  and recall from the Appendix the basic definitions and notations from pathwise functional It\^o calculus. To keep notations simple and close to~\cite{Dupire,CF,Vol15}, we will work in the sequel with a fixed and finite time horizon $T>0$. It is easy to extend our results to the case $T=\infty$ by using localization.  By $D([0,T],U)$ we denote the usual Skorokhod space of $U$-valued c\`adl\`ag functions.  For an open set $V\subset\mathbb R^m$ and a non-anticipative functional $F:[0,T]\times D([0,T],U)\times\CBV([0,T],V)\to\mathbb R$, we denote by 
$$ \mathscr{D}_0F,\dots,\mathscr D_mF$$
 the corresponding horizontal derivatives defined in \eqref{horizontal derivatives} and  by 
 $$ \nabla_XF=(\partial_i F)_{i=1,\dots, d}$$ the vertical derivative defined in \eqref{vertical derivative}, provided that all these derivatives exist.  The second partial vertical derivatives of $F$ will be denoted by $\partial_{ij}F$. We will also need the space $\mathbb{C}^{1,2}_c(U,W)$, whose definition is recalled in the Appendix. The following definition is based on a suggestion in~\cite{Ananova}.

 \begin{definition}\label{locadmcont} Let $U\subset\mathbb R^d$ be an open set. A functional $\xi:[0,T]\times D([0,T],U)\to\mathbb R^d$ is called an \emph{admissible functional integrand on $U$} if  there exist $m\in\mathbb N$, an open set $W\subset\mathbb R^m$, $A\in\CBV([0,T],W)$, and $F\in\mathbb{C}^{1,2}_c(U,W)$  such that $\xi(t,X)=\nabla_XF(t,X,A)$ for $t\in[0,T]$.  \end{definition}

If $\xi$ is an admissible functional integrand on $U\subset\mathbb R^d$ and and $X\in QV^d$ is   $U$-valued, then  the It\^o integral of $\xi(t,X)$ against $X$ can be defined through \eqref{Itointcont}. Moreover, it can be shown that $t\mapsto\int_0^t\xi(s,X)\ud X(s)$ is a continuous function \cite[Lemma 3.1 (c)]{Vol15}.

\begin{definition}An   admissible functional integrand $\pi$ on $\mathbb R^d$ is called a \emph{functional portfolio} if $\pi_1(t,X)+\cdots+\pi_d(t,X)=1$ for all $t\in[0,T]$ and $X\in D([0,T], \mathbb R^d)$ and if $t\mapsto\int_0^t\pi(s,\log S)\ud\log S$ admits the quadratic variation
\begin{equation}\label{path portolio quadr var eq}
\bigg[\int_0\pi(s,\log S)\ud\log S\bigg](t)=\sum_{i,j=1}^d\int_0^t\pi_i(s,\log S)\pi_j(s,\log S)\ud [\log S_i,\log S_j](s)\end{equation}\end{definition}

In contrast to the case without path dependence, the validity of \eqref{path portolio quadr var eq}
 is no longer clear \emph{a priori}  in the path-dependent setting, which is why  \eqref{path portolio quadr var eq} is included as a requirement in the preceding definition. See \cite[Theorem 2.1]{Ananova} for sufficient conditions on $S$ and $\pi$ under which  \eqref{path portolio quadr var eq} holds. 
 
 \begin{rem}\label{associativity remark}If $\pi$ is a functional portfolio, then 
 $$\frac{\pi(t,\log Y)}{Y(t)}:=\Big(\frac{\pi_1(t,\log Y)}{Y_1(t)},\dots, \frac{\pi_d(t,\log Y)}{Y_d(t)}\Big)^\top,\qquad Y\in D([0,T],\mathbb R^d_+),
 $$
 is an admissible functional integrand on $\mathbb R^d_+$ and we have the following change of variables formula:
 $$\int_0^t\pi(s,\log S)\ud\log S(s)=\int_0^t\frac{\pi(s,\log S)}{S(s)}\ud S(s)-\frac12\sum_{i=1}^d\int_0^t\frac{\pi_i(s,\log S)}{(S_i(s))^2}\ud[S_i,S_i](s).
 $$
 This follows from \cite[Corollary 2.1]{Vol15}. In addition, \eqref{log covariation id} and the associativity of the Riemann--Stieltjes integral (see, e.g., \cite[Theorem I.6b]{Widder}) yield  the identity
$$\int_0^t\pi_i(s,\log S)\pi_j(s,\log S)\ud [\log S_i,\log S_j](s)
=\int_0^t\frac{\pi_i(s,\log S)\pi_j(s,\log S)}{S_i(s)S_j(s)}\ud [ S_i, S_j](s).
$$
 \end{rem}

  If $\pi$ is a functional portfolio, we can in analogy to \eqref{log S change eq} and \eqref{ItoDiffEqLSGcadlag} define the corresponding \emph{portfolio value with unit initial wealth} as
 \begin{equation}\label{functional portfolio value}
 \begin{split}
 V^\pi(t)&:=\exp\Bigg(\int_0^t\pi(s,\log S)\ud\log S(s)\\
 &\qquad+\frac12\sum_{i,j=1}^d\int_0^t\Big(\delta_{ij}\pi_i(s,\log S)-\pi_i(s,\log S)\pi_j(s,\log S)\Big)\ud[\log S_i\log S_j](t)\Bigg)\\
 & =\exp\Bigg(\int_0^t\frac{\pi(s,\log S)}{S(s)}\ud S(s)-\frac12\sum_{i,j=1}^d\int_0^t\frac{\pi_i(s,\log S)\pi_j(s,\log S)}{S_i(s)S_j(s)}\ud [ S_i, S_j](s)\Bigg)
  \end{split}
\end{equation}
where $\delta_{ij}$ is the Kronecker delta  and where we have used Remark \ref{associativity remark} for the second identity.
However, in our path-dependent setting, it is not evident that $V^\pi(t)$ can be obtained, in analogy to Lemma~\ref{portfolio value lemma}, as the wealth of a self-financing trading strategy.
For such a statement, a functional extension of $V^\pi$ is needed.

 \begin{lemma}\label{functional portfolio value lemma}
 For a functional portfolio $\pi$, let $W\subset\mathbb R^m$ be open, $F\in\mathbb C^{1,2}_c(\mathbb R^d,W)$, and $A\in\CBV([0,T],W)$ be such that $\pi(t,X)=\nabla_XF(t,X,A)$. Define for $X\in D([0,T],\mathbb R^d_+)$
 $$\bbar V^\pi(t,X,(A,B^\pi)):=\exp\Big(F(t,\log X,A)-F(0,\log X,A)-B^\pi(t)\Big),
 $$
 where, for $A_0(t)=t$, 
  \begin{align}
 B^\pi(t)&=\sum_{i=0}^m\int_0^t\mathscr{D}_iF(s,\log S,A)\ud A_i( s)\label{Bpi eq}\\
 &\quad + \frac{1}{2}\sum_{i,j=1}^d\int_0^t\Big( (\partial_{ij} F)(s, \log S  ,A)+\pi_i(s,\log S)\pi_j(s,\log S)-\delta_{ij}\pi_i(s,\log S)\Big) \ud [\log S_i,\log S_j](s).\nonumber
  \end{align}
  Then $\bbar V^\pi\in\mathbb C^{1,2}_c(\mathbb R^d_+,W\times\mathbb R)$, the functional $\xi:[0,T]\times D([0,T],\mathbb R^d_+)\to\mathbb R^d$ defined through
  $$\xi_i(t,X)=\frac{\pi_i(t,\log X)V^\pi(t,X,(A,B))}{X_i(t)},\qquad i=1,\dots, d,
  $$
  is an admissible integrand on $\mathbb R^d_+$, and 
  \begin{equation}\label{funct portf sf cond}
  \bbar V^\pi(t,S,(A,B^\pi))=1+\int_0^t\xi(s,S)\ud S(s)=V^\pi(t)\qquad\text{for $0\le t\le T$.}
  \end{equation}
   \end{lemma}

We will also need the following functional extension of the market portfolio.
For $X\in D([0,T],\mathbb R^d)$, let
\begin{equation}\label{functional market portfolio}
\bbar\mu_i(t,X):=\frac{e^{X_i(t)}}{e^{X_1(t)}+\cdots+e^{X_d(t)}}, \qquad i=1,\dots, d.
\end{equation}
Then $\bbar \mu(t,X)\in\Delta^d_+$, and $\bbar\mu(t,\log S)$ is equal to the market portfolio $\mu(t)$ from Lemma~\ref{market portfolio lemma}. 

\begin{lemma}\label{functional market portfolio lemma} The functional $\bbar\mu$ is a functional portfolio, and
$$\bbar V^{\bbar\mu}(t,X,B^{\bbar\mu})=\frac{e^{X_1(t)}+\cdots+e^{X_d(t)}}{e^{X_1(0)}+\cdots+e^{X_d(0)}},\qquad X\in D([0,T],\mathbb R^d).
$$
\end{lemma}

The following lemma extends Lemma~\ref{G generated portfolio lemma} to our present functional setting.

\begin{lemma}\label{pathdependent functionally generated portfolio}Let $U\supset \Delta^d_+$ and $W\in\mathbb R^m$ be open sets. For strictly positive $G\in\mathbb C^{1,2}_c(U,W)$ and  $A\in\CBV([0,T],W)$ we let 
\begin{equation}
g_i(t,X):=\partial_i\log G(t,X,A).
\end{equation}
Then
\begin{equation}\label{functional portfolio generated by G eq}
\pi_i(t,X):=\bbar\mu_i(t,X)\Big(1+g_i(t,\bbar\mu(\cdot,X))-\sum_{j=1}^d\bbar\mu_j(t,X)g_j(t,\bbar\mu(\cdot,X))\Big)
\end{equation}
is an admissible integrand on $\mathbb R^d$ satisfying $\pi_1+\cdots+\pi_d=1$.
\end{lemma}

If $\pi$ defined by \eqref{functional portfolio generated by G eq} satisfies  \eqref{path portolio quadr var eq}, then it is a functional portfolio and will be called the \emph{functional portfolio generated by $G$.}
Now we can state the path-dependent version of Theorem~\ref{MasterFormulaCadlag}.

\begin{thm}\label{Path-dependent pathwise master formula}{\bf(Path-dependent pathwise master formula)} Let $G$ and $\pi$ be as in Lemma~\ref{pathdependent functionally generated portfolio} and assume that \eqref{path portolio quadr var eq} holds. Let furthermore  $V^\pi$ be as in \eqref{functional portfolio value} and $V^{\mu}$ be the portfolio value \eqref{market portfolio value eq} of the market portfolio $\mu$. Then 
\begin{equation*} 
\log \left( \dfrac{V^{\pi}(T)}{V^{\mu}(T)} \right) = \log \left( \dfrac{G( T,\mu,A)}{G( 0,\mu,A)} \right) + \mathfrak{g}([0,T])+ \mathfrak{h}([0,T])  , \quad 0 \leq T < \infty,
\end{equation*}
where, for $A_0(t):=t$, 
\begin{equation*}
\mathfrak{g}(\mathrm{d}t) := -\dfrac{1}{2} \sum_{i,j=1}^{d} \frac{ \partial_{ij}G( t,\mu,A)}{G(t, \mu,A)}  \ud [\mu_i,\mu_j](t)\quad\text{and}\quad \mathfrak{h}(\mathrm{d}t) := -\sum_{k=0}^m \dfrac{\mathscr{D}_kG(t,\mu,A)}{G(t,\mu,A)}    \ud A_k(t).\end{equation*}
\end{thm}

\section{Examples}\label{Examples section}

In this section, we will discuss a class of examples, along with an empirical analysis. The general idea is to use portfolio-generating functions of a convex combination of the market portfolio $\mu(t)$ and its moving average defined by
$$\alpha(t)=\frac1\theta\int_{t-\theta}^t\mu(0\vee s)\ud s,
$$
where $\theta>0$ is given. The portfolio-generating function is then of the form $\varphi (\lambda \mu(t)+(1-\lambda)\alpha(t))$, where $\varphi $ is a strictly positive and twice continuously differentiable function defined on an open and convex neighborhood $U$ of $\Delta^d_+$. It can be considered within the contexts of both Section~\ref{non-pathdep setting section} and Section~\ref{path-dependent master formula section}. In the context of the first section, $\alpha(t)$ may be be regarded as an additional component of bounded variation, and so we may work with the function
\begin{equation}\label{non-pathdep ex}
G(\mu(t),\alpha(t))=\varphi \big(\lambda \mu(t)+(1-\lambda)\alpha(t)\big).
\end{equation}
In the context of Section~\ref{path-dependent master formula section}, however, we may consider a non-anticipative functional $\wh G$ that depends on the \emph{path} of $\mu$ via $\alpha$. To make this latter idea precise, we let for $X\in D([0,T],U)$,
$$\wh G(t,X):=\varphi \Big(\lambda X(t)+\frac{1-\lambda}\theta\int_{t-\theta}^tX(0\vee s)\ud s\Big).
$$
Then $\wh G\in \mathbb C^{1,2}_c(U,\emptyset)$ and
\begin{equation}\label{pathdep ex}
G(t,\mu(t))=\wh G(t,\mu).
\end{equation}
The following proposition states in particular that the two descriptions \eqref{non-pathdep ex} and \eqref{pathdep ex} lead to the same results.

\begin{prop}\label{Examples eq prop} Let $\pi(t)$ be the portfolio generated by the function $G$ in \eqref{non-pathdep ex} and $\wh\pi(t,X)$ as in Lemma~\ref{pathdependent functionally generated portfolio} for $\wh G$ from \eqref{pathdep ex}. Then, for
$$g_i(t):=\frac{\lambda \varphi _{x_i}\big(\lambda \mu(t)+(1-\lambda)\alpha(t)\big)}{\varphi \big(\lambda \mu(t)+(1-\lambda)\alpha(t)\big)},\qquad i=1,\dots, d,
$$
we have
\begin{equation}\label{examples general portfolio eq}
\wh\pi_i(t,\log S)=\pi_i(t)=\mu_i(t)\Big(1+g_i(t)-\sum_{j=1}^d\mu_j(t)g_j(t)\Big).
\end{equation}
Moreover, $\wh\pi$ satisfies \eqref{path portolio quadr var eq} and is hence a functional portfolio. Furthermore, we have  $V^\pi(T)=V^{\wh \pi}(T)$ and
\begin{equation}\label{master eq ex eq}
\log\bigg(\frac{V^\pi(T)}{V^\mu(T)}\bigg)=\frac{\varphi \big(\lambda \mu(T)+(1-\lambda)\alpha(T)\big)}{\varphi \big(\lambda \mu(0)+(1-\lambda)\alpha(0)\big)}+ \mathfrak{g}([0,T])+ \mathfrak{h}([0,T]),
\end{equation}
where 
\begin{equation}
\begin{split}
\mathfrak{h}([0,T])&=-\frac{1-\lambda}{ \lambda}\sum_{i=1}^d\int_0^Tg_i(t)\alpha_i'(t)\ud t\\
\mathfrak{g}([0,T])&=- \frac{\lambda^2}2\sum_{i,j=1}^{d}  \int_0^T\frac{\varphi_{x_i,x_j} (\lambda \mu(t)+(1-\lambda)\alpha(t))}{ \varphi (\lambda \mu(t)+(1-\lambda)\alpha(t))} \ud [\mu_i,\mu_j](t)
\end{split}
\end{equation}
\end{prop}

In the following examples, we analyze  empirically the situation for particular choices for $\varphi$. To this end, it will be convenient to use the shorthand notation
$$\widetilde\mu_i(t)=\lambda\mu_i(t)+(1-\lambda)\alpha_i(t)\qquad i=1,\dots, d.
$$

\begin{ex}[\bf{Geometric average}]\label{GeoMeanExample}
Consider $\varphi(x)=\prod_{i=1}^dx_i^{1/d}$, 
 which generates the portfolio 
\begin{align}
\pi_{i}(t) &= \left(1+ \dfrac{\lambda}{d \widetilde{\mu}_{i}(t)}  - \sum_{j = 1}^{d} \dfrac{\lambda \mu_{j}(t)}{d \widetilde{\mu}_{j}(t)}  \right) \mu_{i}(t),\quad i=1,\dots, d.  
\label{NL4-weights}
\end{align}
 By Proposition~\ref{Examples eq prop}, we have, with $\delta_{ij}$ denoting again the Kronecker delta,
\begin{align*}
\mathfrak{g}([0,T])&=\dfrac{\lambda^{2}}{2d}\sum_{i,j=1}^d  \int_0^T\left(\sum_{i=1}^{d} \dfrac{\delta_{ij}}{\left( \widetilde{\mu}_{i}(t) \right)^{2}}  -  \dfrac{1}{d\widetilde{\mu}_{i}(t)\widetilde{\mu}_{j}(t) } \right)\ud[\mu_i,\mu_j](t)\\
\mathfrak{h}([0,T])&=-\dfrac{1-\lambda}{d }\sum\limits_{i=1}^{d}\int_0^T\frac{\alpha_i'(t)}{  \widetilde{\mu}_{i}(t) }\ud t.
\end{align*}
In Figures~\ref{fig:NL4_1} and~\ref{fig:NL4_22}, we display the results of an empirical analysis of such a geometrically weighted portfolio with the parameters $\theta = 60$ days and  $\lambda = 0,7$. We used the stock data base from Reuters Datastream and  considered those 30  stocks that were the constituents of the DAX in May 2015. Our 30 time series represent daily closing prices for those stocks during the period January 31, 2005 -- May 15, 2015,  and we work along the finite partition formed by the corresponding time points. 
In Figure~\ref{fig:NL4_1} we see the relative performance  of the portfolio \eqref{NL4-weights} with respect to this stock index. 
In  Figure~\ref{fig:NL4_22} we see the decomposition of the curve(s) in the left-hand panel 
according to the master equation \eqref{master eq ex eq}. The blue curve is the change in the generating functional, while the red and the green ones are the respective drift terms.  Each curve shows the cumulative value of the  daily changes induced in the corresponding quantities by capital gains and losses. As can be seen, the cumulative second-order drift term  was the dominant part over the period, with a total contribution of about 15 percentage points to the relative return.
The second-order drift term was quite stable over the considered period, with an  exception during the period around the financial crisis of 2008. 

       \begin{minipage}{8.5cm}
          \begin{figure}[H]
               \includegraphics[width=8.5cm]{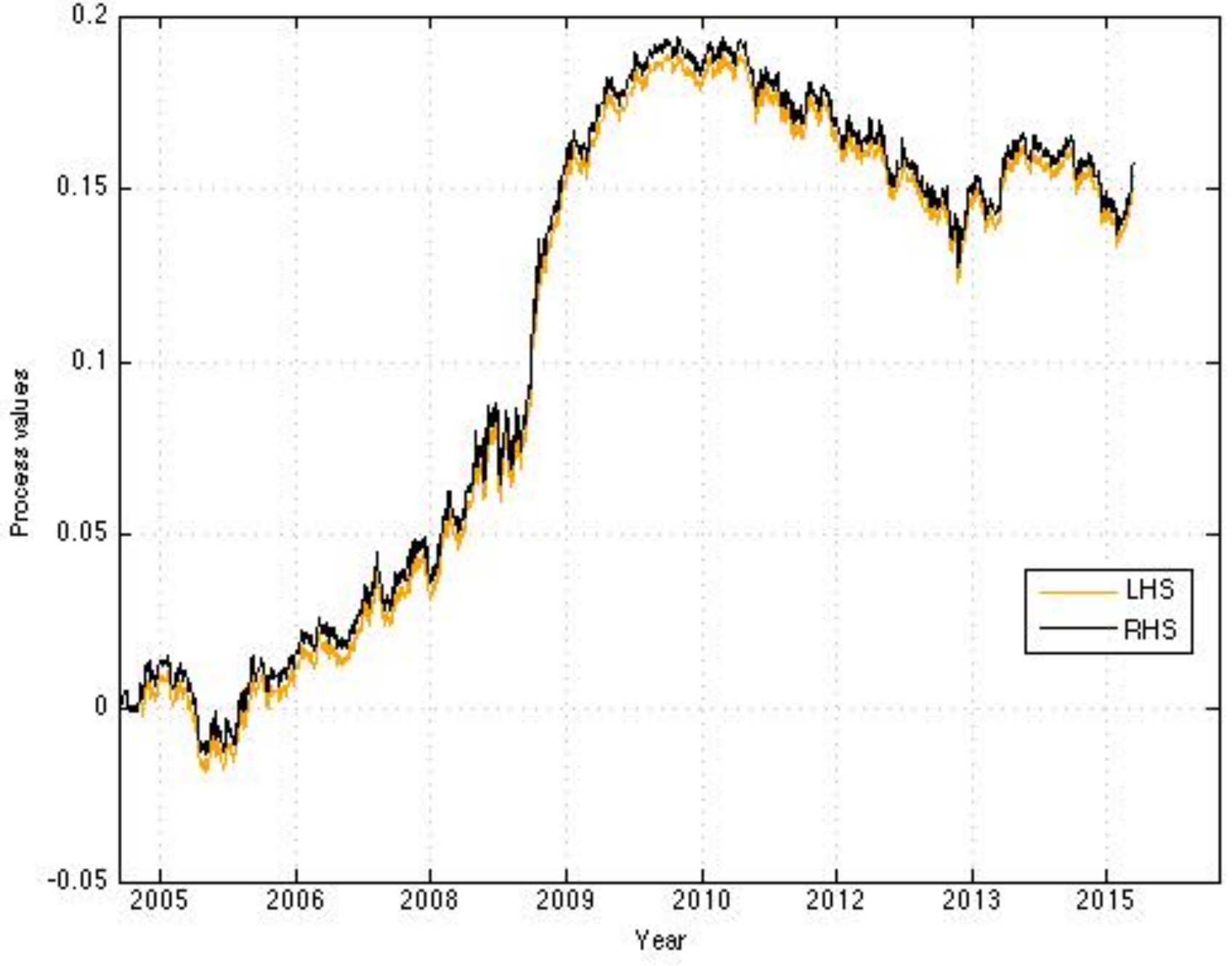} 
                \caption{LHS vs.~RHS of the master formula \eqref{master eq ex eq} for geometric average} 
                \label{fig:NL4_1} 
          \end{figure}
      \end{minipage}
      \hspace{0.05\linewidth}
      \begin{minipage}{8.5cm}
          \begin{figure}[H]
                \includegraphics[width=8.5cm]{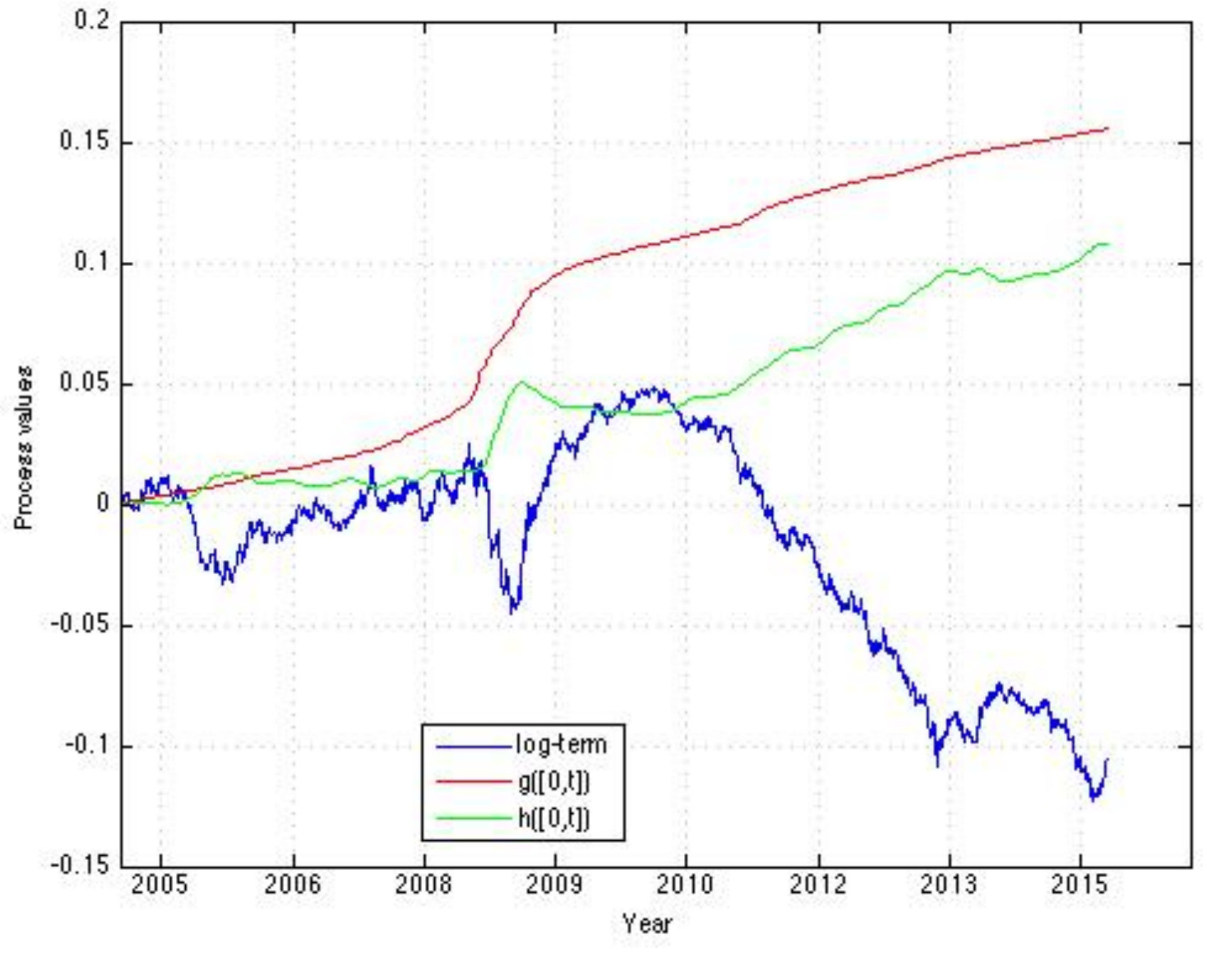} 
                \caption{Componentwise representation of the RHS of  \eqref{master eq ex eq} for geometric average} 
                \label{fig:NL4_22} 
          \end{figure}
      \end{minipage}
\end{ex}

\begin{ex}[\bf{Functional diversity weighting}]\label{DiversityWeightingExample}
Here we take $p\in(0,1)$ and
$\varphi(x)=\big(\sum_{i=1}^dx_i^p\big)^{1/p}$,
which generates the portfolio with weights
\begin{align}
\pi_{i}(t) &= \left(1+ \dfrac{\lambda \left(\widetilde{\mu}_{i}(t) \right)^{p-1} }{\sum_{k=1}^{d} \left(\widetilde{\mu}_{k}(t) \right)^{p}}  - \sum_{j = 1}^{d} \dfrac{\lambda \mu_{j}(t) \left(\widetilde{\mu}_{j}(t) \right)^{p-1}}{\sum_{k=1}^{d} \left(\widetilde{\mu}_{k}(t) \right)^{p}}  \right)\mu_{i}(t),\qquad i=1,\dots, d.   
\label{Dalpha_1-weights}\end{align}
By Proposition~\ref{Examples eq prop}, we have
\begin{align*}
\mathfrak{g}([0,T])&=\dfrac{\lambda^2 (1-p)}{2} \sum_{i,j=1}^d\int_0^T\bigg( \dfrac{ \delta_{ij}\left(\widetilde{\mu}_i(t)\right)^{p-2}}{\sum_{k=1}^{d} \left(\widetilde{\mu}_k(t)\right)^{p}} - \dfrac{  \left(\widetilde{\mu}_{i}(t) \right)^{p-1} \left(\widetilde{\mu}_{j}(t) \right)^{p-1}}{\left(\sum_{k=1}^{d} \left(\widetilde{\mu}_{k}(t) \right)^{p}\right)^{2}}  \bigg) \ud[\mu_i,\mu_j](t)\\
\mathfrak{h}([0,T])&=-\int_0^T \dfrac{1-\lambda}{ \sum_{k=1}^{d} \left(\widetilde{\mu}_{k}(t)\right)^{p}} \sum\limits_{i=1}^{d}\left( \widetilde{\mu}_{i}(t) \right)^{p-1}\alpha_i'(t)\ud t  .
\end{align*}
For our empirical analysis for the diversity-weighted portfolio, we used again Reuters Datastream to obtain our data base;  we considered 207 fixed stocks that were among the constituents of the S\&P 500 index in May 2015. Our 207 time series represent monthly average prices for the period February 2, 1973 -- April 2, 2015. 
The results of a simulation of the portfolio  \eqref{Dalpha_1-weights} using the parameters $\theta = 12$ months, $\lambda = 0,6$, and $p = 0,1$ are presented  below. Figure~\ref{fig:Dalpha_1} shows the relative performance of this portfolio with respect to this filtered index, and Figure~\ref{fig:Dalpha_22} shows its decomposition  in the three components according to the master equation \eqref{master eq ex eq}. Each curve   represents the cumulative value of the monthly changes induced in the corresponding quantities by capital gains and losses, but in contrast to Example~\ref{GeoMeanExample}, it is now the cumulative change in the generating functional  that yields the dominant part over the period, with a total contribution of about 70 percentage points to the relative return. The second-order drift term was quite stable over the period with a  total contribution of about 30 percentage points, whereas the horizontal drift term has a negative contribution.

\noindent \begin{minipage}{\linewidth}
      \centering
      \begin{minipage}{0.45\linewidth}
          \begin{figure}[H]
                \centering 
                 \includegraphics[width=8.5cm]{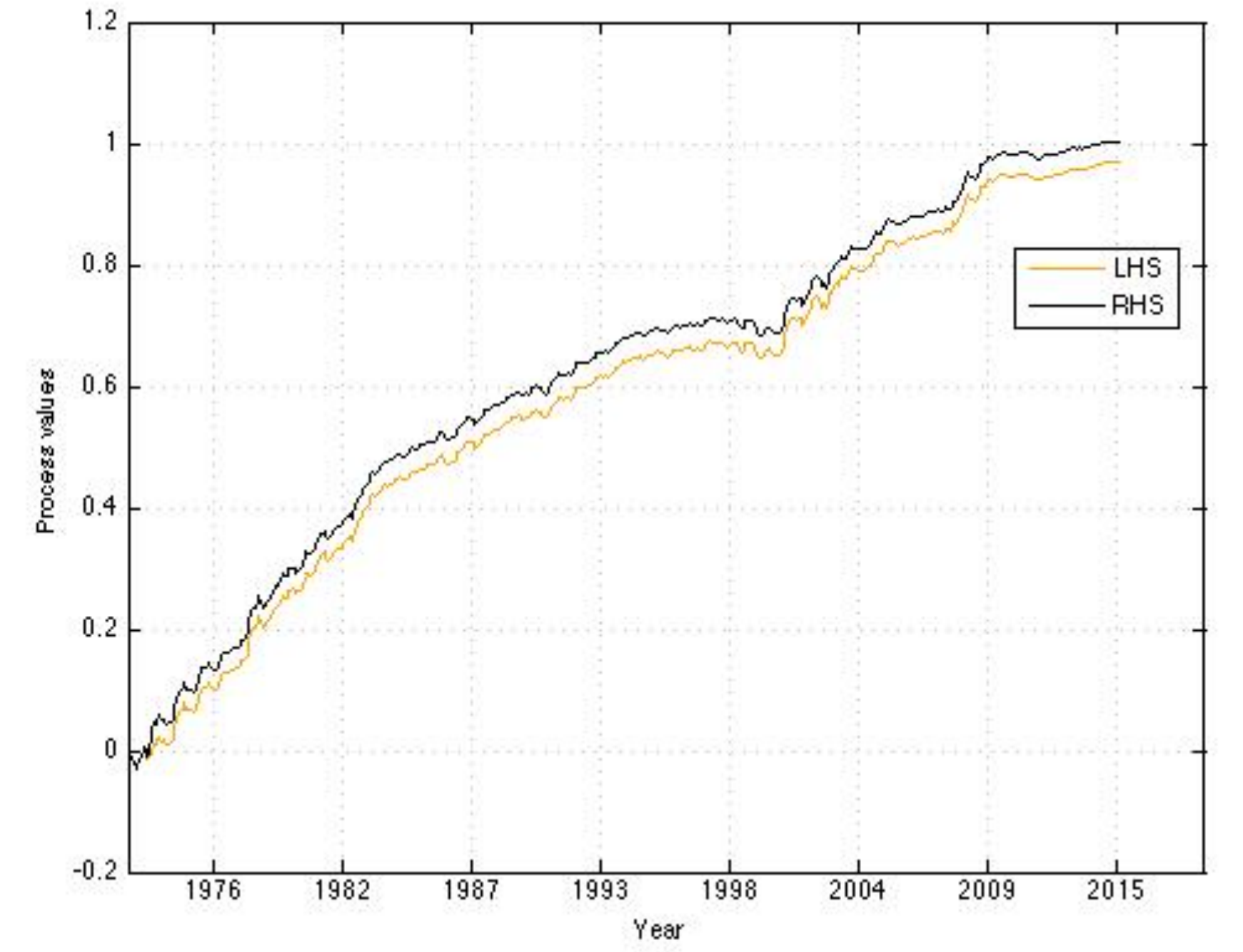} 
                \caption{LHS vs. RHS of the master formula  \eqref{master eq ex eq} for diversity weighting} 
                \label{fig:Dalpha_1} 
          \end{figure}
      \end{minipage}
      \hspace{0.05\linewidth}
      \begin{minipage}{0.45\linewidth}
          \begin{figure}[H]
          	    \centering 
                \includegraphics[width=8.5cm]{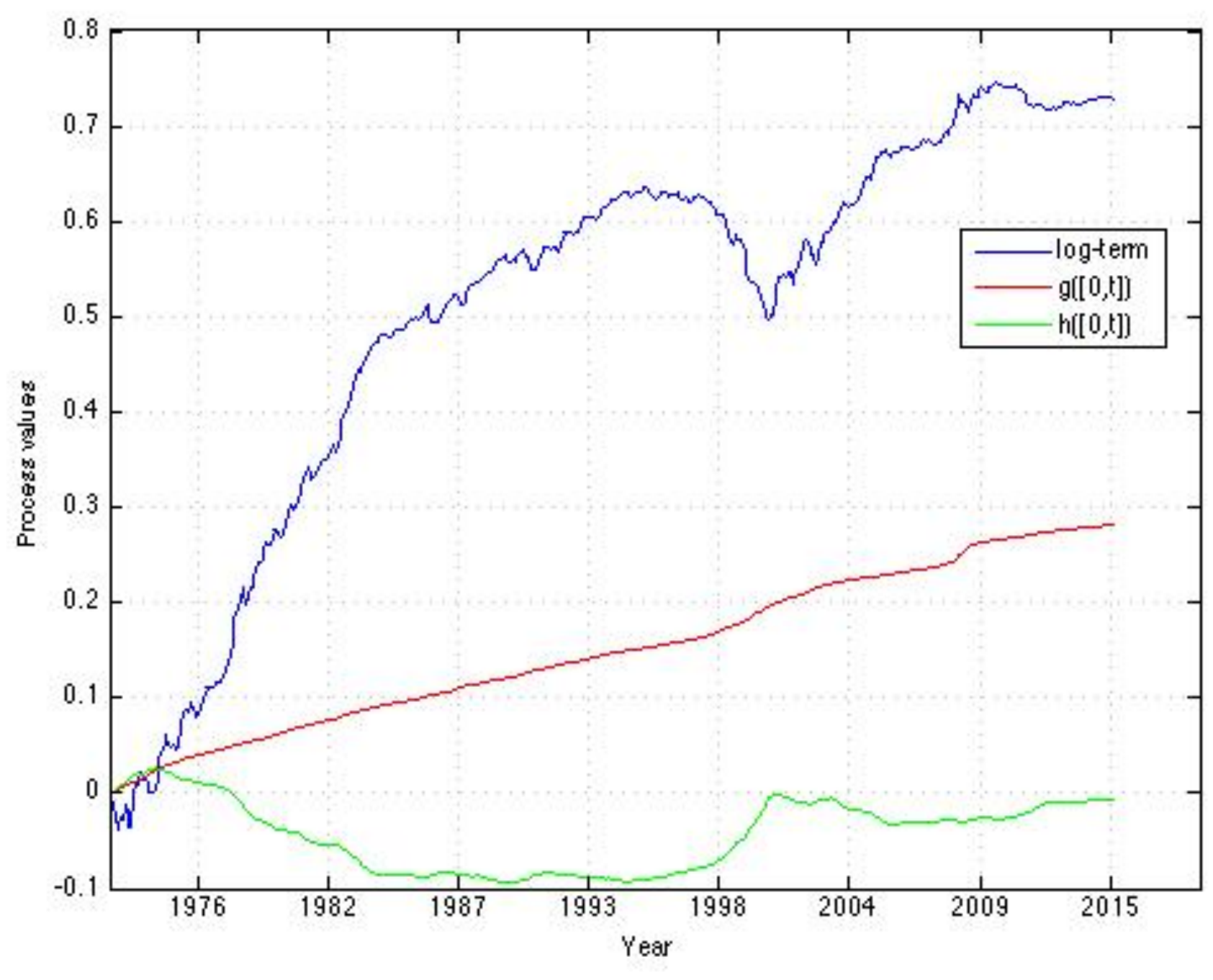} 
                \caption{Componentwise representation of the RHS of  \eqref{master eq ex eq} for diversity weighting} 
                \label{fig:Dalpha_22} 
          \end{figure}
      \end{minipage}
  \end{minipage}
\end{ex}

\begin{ex}[\bf{Functional entropy-weighting}]\label{Entro} 
Here we take $\varphi(x)=-\sum_{i=1}^dx_i\log x_i$, 
which generates the portfolio with weights
\begin{align}
\pi_{i}(t) &= \left(1- \dfrac{\lambda \log \left(\widetilde{\mu}_{i}(t) \right)}{\varphi( \widetilde{\mu}(t))}+ \sum_{j = 1}^{d} \dfrac{\lambda \mu_{j}(t) \log \left(\widetilde{\mu}_{j}(t) \right)}{\varphi( \widetilde{\mu}(t) )}  \right) \mu_{i}(t),   
\label{minus-Halpha_1-weights}
\end{align}
\noindent and associated drift rates
\begin{align*}
\mathfrak{g}([0,T])&=-\lambda^2 \sum_{i=1}^{d} \int_0^T \dfrac{1}{ \varphi(\wt\mu(t))\widetilde{\mu}_{i}(t)} \ud[\mu_i,\mu_i](t),\\
\mathfrak{h}([0,T])&=(1-\lambda)\sum_{i=1}^d\int_0^T\frac{1+\log\wt\mu_i(t)}{\varphi(\wt\mu(t))}\alpha_i'(t)\ud t.
\end{align*}
Using the same data set as in Example~\ref{DiversityWeightingExample}, we conducted an empirical analysis of the portfolio \eqref{minus-Halpha_1-weights} with respect to the filtered index taking the parameters $\theta = 6$ months and  $\lambda = 0,9,$ which is presented in Figure~\ref{fig:Halpha_1} and Figure~\ref{fig:Halpha_22}, respectively.
Note that compared to the entropy-weighted portfolio in the classical Fernholz' setting  (see \cite[Example11.1]{FK08})
 the additional drift term $\mathfrak{h}([0,T])$  can be either positive or negative.
In this sense,  the horizontal drift term $\mathfrak{h}([0,T])$  may be regarded as quantifying the {trade-off} between outperforming the market and the greater flexibility within the more general framework.
This is also supported by real market data, as can be seen in Figure~\ref{fig:Halpha_1}  and Figure~\ref{fig:Halpha_22}. Indeed,   the cumulative second-order drift term  is continually increasing. Moreover,  the horizontal drift term does not seem to have a large influence on the relative performance of the entropy-weighted portfolio, with a total contribution of less than 1 percentage point. Thus, entropy-weighting should significantly outperform the market on the considered time interval, which is confirmed in the case study of Figure~\ref{fig:Halpha_1}.

\noindent \begin{minipage}{\linewidth}
      \centering
      \begin{minipage}{0.45\linewidth}
          \begin{figure}[H]
                \centering 
              \includegraphics[width=8.5cm]{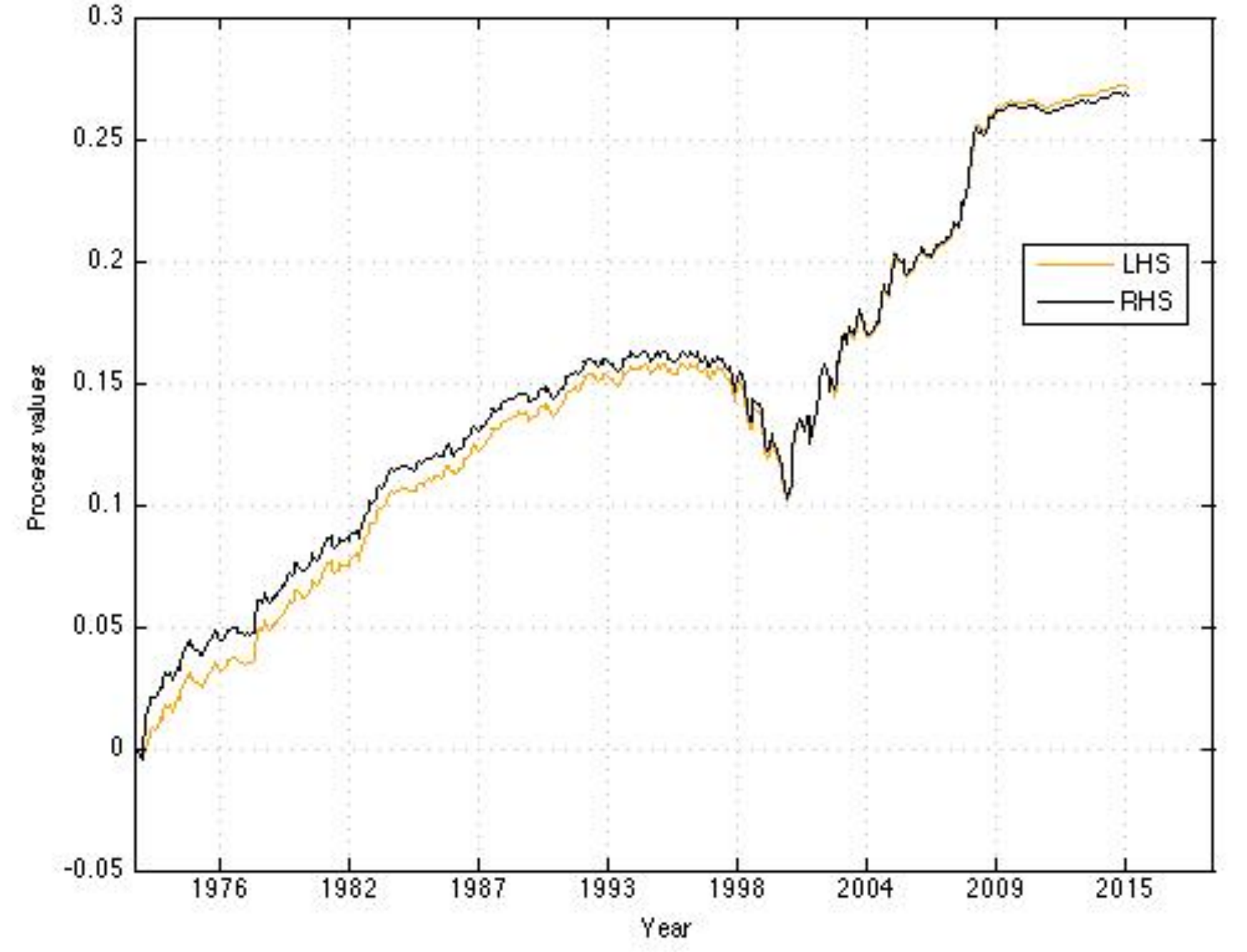} 
                \caption{LHS vs. RHS of the master formula  \eqref{master eq ex eq} for entropy weighting} 
                \label{fig:Halpha_1} 
          \end{figure}
      \end{minipage}
      \hspace{0.05\linewidth}
      \begin{minipage}{0.45\linewidth}
          \begin{figure}[H]
          	    \centering 
              \includegraphics[width=8.5cm]{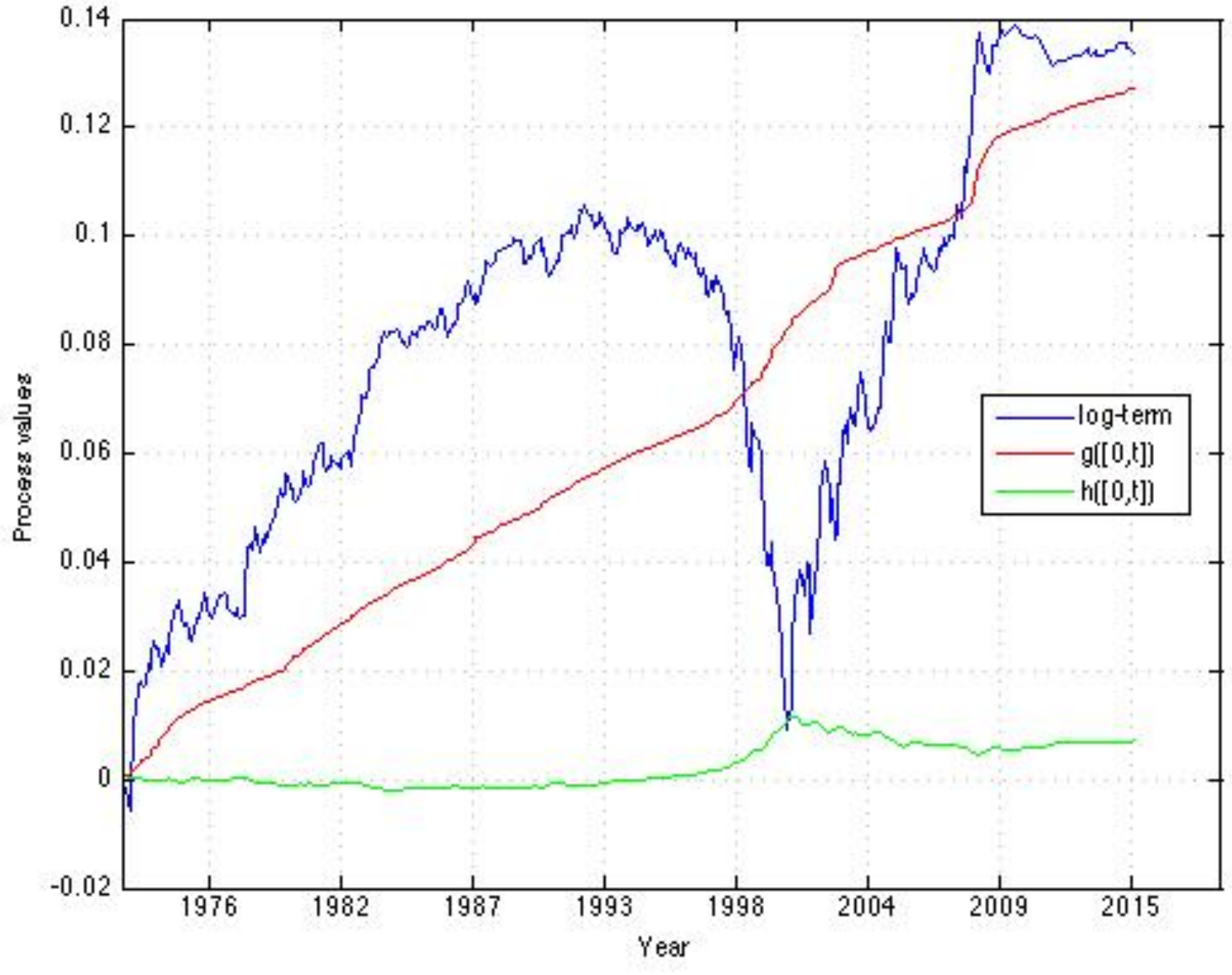} 
                \caption{Componentwise representation of the RHS of  \eqref{master eq ex eq} for entropy weighting} 
                \label{fig:Halpha_22} 
          \end{figure}
      \end{minipage}
  \end{minipage}
\vspace*{10pt}\\

\ignore{\red{\sout{ Note also that the above example is valid in a very general context, since we have not imposed in the above discussion \emph{any} assumption on the volatility structure of the market model   beyond the absolutely minimal condition --  that the stock price vector $S$ should admit continuous quadratic variation.}}}
 \end{ex}

 \section{Proofs}\label{Proofs section}
\subsection{ Proofs of the lemmas from Section~\ref{non-pathdep setting section}} \label{Sec3}

 \begin{proof}[Proof of Lemma~\ref{log lemma}] We first  let $Z:[0,\infty)\to\R_+$ be any continuous path. It follows from \cite[Proposition 2.2.10]{Sondermann} that $Z\in QV^1_+$ if and only $\log Z\in QV^1$ and that, in this case,
\begin{equation}\label{simple log identity}
[\log Z](t)=\int_0^t\frac1{Z(s)^2}\ud [Z](s).
\end{equation}
Now suppose that  $S\in QV^d_+$. By the polarization identity \eqref{polarization identity}, we may conclude that $\log S\in QV^d$, if we can show that $f(S_i(t),S_j(t))$ belongs to $QV^1$ for all $i,j$, where $f(x_1,x_2):=\log x_1+\log x_2$. The pathwise It\^o formula from~\cite{Ito_F}
applied to $f(S_i(t),S_j(t))$ yields that
\begin{align*}
f(S_i(t),S_j(t))&=f(S_i(0),S_j(0))+\int_0^t\nabla f(S_i(s),S_j(s))\ud{S_i(s)\choose S_j(s)}\\&\quad+\frac12\sum_{k,\ell=1}^2\int_0^t f_{x_k,x_\ell}(S_i(s),S_j(s))\ud[S_k,S_\ell](s).
\end{align*}
 Remark 8 and Proposition 12 from~\cite{Schied13} now imply that $f(S_i,S_j)\in QV^1$ with quadratic variation
\begin{align*}[f(S_i,S_j)](t)&=\sum_{k,\ell=1}^2\int_0^tf_{x_k}(S_i(s),S_j(s))f_{x_\ell}(S_i(s),S_j(s))\ud [S_k,S_\ell](s)\\
&=\int_0^t\frac1{S_i(s)^2}\ud[S_i](s)+\int_0^t\frac1{S_j(s)^2}\ud[S_j](s)+2\int_0^t\frac1{S_i(s)S_j(s)}\ud [S_i,S_j](s).
\end{align*}
Therefore, $\log S\in QV^d$ and \eqref{log covariation id}
 now follow from \eqref{polarization identity} and \eqref{simple log identity}. That $\log S\in QV^d$ implies $S\in QV^d_+$ follows by an analogous argument in which the logarithm is replaced by the exponential function. 
\end{proof}

\begin{proof}[Proof of Lemma~\ref{pi lemma}]Let $X:=\log S$ and write $\exp(X)$ for the path with components $e^{X_1},\dots, e^{X_d}$. That is, $\exp(X)=S$. Now suppose that $\pi$ is an admissible integrand for $X$ and $T>0$ is given. Then there exist  $n\in\bN$,  open sets $U\subset\mathbb R^d$ and   $V\subset\mathbb R^n$,  a function $f\in C^{2,1}(U,V)$, and  $ A\in\CBV([0,T],V)$  such that $X(t)\in U$  and  $\xi(t)=\nabla_{  x}f(X(t),A(t))$ for $0\le t\le T$. We let $g(z,a):=f(\log z,a)$. Then $g\in C^{2,1}(\widetilde U,V)$ for the open set $\widetilde U=\{\exp(x)\,|\,x\in U\}$ and  $g_{z_i}(z,a)=f_{x_i}(\log z,a)/z_i$. Thus, $g_{z_i}(S(t),A(t))=\pi_i(t)/S_i(t)$, $i=1,\dots, d$, are the components of an admissible integrand for $S$. The converse assertion follows analogously. To prove the formula \eqref{log S change eq}, we use F\"ollmer's pathwise It\^o formula,  which yields 
\begin{equation}\label{log S Ito eq}
\ud \log S_i(t)=\frac1{S_i(t)}\ud S_i(t)-\frac1{2 S_i(t)^2}\ud[S_i](t).
\end{equation}
Thus, the associativity rule   \cite[Theorem 13]{Schied13} implies \eqref{log S change eq}.
\end{proof}

\begin{proof}[Proof of Lemma~\ref{portfolio value lemma}] It is clear from Lemma~\ref{pi lemma} that the It\^o integral in \eqref{ItoDiffEqLSGcadlag} is well defined. Moreover,   the two rightmost integrals in \eqref{ItoDiffEqLSGcadlag} exist as Riemann--Stieltjes integrals, because $\pi(t)/S(t)$ as admissible integrand for $S$ is in particular a continuous function of $t$.
Therefore, 
$V^\pi$ is well defined and indeed the wealth of $(\xi,\eta)$, since $V^\pi(t)=\xi(t)\cdot S(t)+\eta(t)B(t)$ by \eqref{handelsstrXIETA}. Letting $Z(t):=\int_0^t\frac{\pi(s)}{S(s)}\ud S(s)$ and $R(t):=\int_0^t(1-\sum_{i=1}^d\pi(s))r(s)\ud s$, \eqref{ItoDiffEqLSGcadlag}  can be re-written as
$V^\pi(t)=\exp\big(Z(t)-\frac12[Z](t)+R(t)\big)$. Applying the pathwise It\^o formula thus gives
$\ud V^\pi(t)=V^\pi(t)\ud Z(t)+V^\pi(t)\ud R(t)$. The associativity rule \cite[Theorem 13]{Schied13}  thus implies that $V^\pi(t)\pi(t)/S(t)=\xi(t)$ is an admissible integrand for $S$ and that $V^\pi(t)\ud Z(t)=\xi(t)\ud S(t)$. Moreover, the identity $V^\pi(t)\ud R(t)=\eta(t)\ud B(t)$ follows from the associativity of the Riemann--Stieltjes integral (see, e.g., \cite[Theorem I.6b]{Widder}).
\end{proof}

\begin{proof}[Proof of Lemma~\ref{market portfolio lemma}] 
For $x=(x_1,\dots, x_d)$ let $h(x):=\log (e^{x_1}+\cdots+e^{x_d})$. Then $h_{x_i}(x)=e^{x_i-h(x)}$ and so $\nabla h(\log S(t))=\mu(t)$. Thus, $\mu$ is an admissible integrand for $\log S$. Moreover, the fact that $\mu_1(t)+\cdots+\mu_d(t)=1$ is obvious. Therefore, $\mu$ is a portfolio in the sense of Definition~\ref{portfolio}.

To prove the formula for $V^\mu$, let $g(x_1,\dots, x_d):=\log(x_1+\cdots+x_d)$ for $x_i>0$.  The pathwise It\^o formula  yields
$$g(S(t))-g(S(0))=\int_0^t\frac{\mu(s)}{S(s)}\ud S(s)-\dfrac{1}{2} \sum\limits_{i,j = 1}^{d} \int_0^t \dfrac{\mu_{i}(s) \mu_{j}(s)}{S_{i}(s)S_{j}(s)} \,\mathrm{d} [ S_{i}, S_{j} ] (s).
$$
By Lemma~\ref{portfolio value lemma}, the right-hand side is equal to $\log V^\mu(t)-\log V^\mu(0)$, which proves the claim.
\end{proof}

\begin{proof}[Proof of Lemma~\ref{G generated portfolio lemma}]  We clearly have $\pi_1(t)+\cdots+\pi_d(t)=1$. We  show next that $\pi$ is an admissible integrand for $\log S$ and, hence, a portfolio for $S$. To this end, let $h$ be as in the proof of Lemma~\ref{market portfolio lemma}, so that $\mu(t)=\nabla h(X(t))$, where again $X(t)=\log S(t)$. Observe that $\nabla_xh(x)\in\Delta^d_+$ for all $x\in\mathbb R^d$. Therefore, we may define
$$\gamma (x,a):=\log G(\nabla h(x),a),\qquad x\in\mathbb R^d,\, a\in W.
$$
Then
$$\gamma _{x_i}(x,a)=\sum_{j=1}^d\Big(\frac{\partial}{\partial{x_j}}\log G\Big)(\nabla h(x),a)h_{x_i,x_j}(x).
$$
Since
$$h_{x_i,x_j}(X(t))=\delta_{ij}e^{{X_i(t)}-h(X(t))}-e^{X_i(t)+X_j(t)-2h(X(t))}=\delta_{ij}\mu_i(t)-\mu_i(t)\mu_j(t),$$
we get that
\begin{equation}\label{widetilde pi grad identity}
\widetilde\pi(t):=\gamma _{x_i}(X(t),A(t))=\mu_i (t)\frac{\partial}{\partial{x_i}} \log G(\mu(t),A(t))  - \sum_{j = 1}^{d} \mu_i(t)\mu_{j}(t) \frac{\partial}{\partial{x_j}}\log G(\mu(t), A(t))
\end{equation}
 is an admissible integrand for $\log S$. With Lemma~\ref{market portfolio lemma} we can conclude now that $\pi=\mu+\widetilde\pi$ is an admissible integrand for $\log S$.

Now we prove that $\pi$ is an admissible integrand for $\log\mu$. To this end, observe first that $h(\log\mu(t))=\log\sum_{i}\mu_i(t)=0$. Hence, $h_{x_i}(\log\mu(t))=e^{\log\mu_i(t)-h(\log\mu(t))}=\mu_i(t)$, and it follows that  $\nabla h(\log S(t))=\mu(t)=\nabla h(\log \mu(t))$. Therefore, the identity \eqref{widetilde pi grad identity} also holds for $X:=\log\mu$, and it follows that $\widetilde\pi$ is an admissible integrand also for $\log\mu$. 
Hence, it only remains to  
prove that $\mu$ is an admissible integrand for $\log\mu$. To this end, let $f(x):=e^{h(x)}$ and note that $f_{x_i}(x)=e^{x_i}$. Thus, $\nabla f(\log\mu(t))=\mu(t)$, and so $\mu$ is indeed an admissible integrand for $\log\mu$.  \end{proof}

\subsection{Pathwise portfolio dynamics without path dependence}

In this section, we analyze the dynamics of portfolios and portfolio values. In standard stochastic portfolio theory, many of these results are well known (see, e.g.,~\cite{FK08}), but occasionally additional care is needed  in our pathwise setting. The  results in this section will serve as preparation for the proof of the functional master formula, Theorem~\ref{MasterFormulaCadlag}, but they are also of independent interest.
 
 Throughout this section, $\pi$ will be a portfolio for $S\in QV^d_+$ in the sense of Definition~\ref{portfolio} and $V^\pi$ the corresponding portfolio value with unit initial wealth  as given in  \eqref{ItoDiffEqLSGcadlag}. Due to  the condition that $\pi^1(t)+\cdots+\pi^d(t)=1$, formula  \eqref{ItoDiffEqLSGcadlag} simplifies to
 \begin{equation}\label{simple portfolio value eq}
 V^{\pi} (t)=\exp \left(   \int_0^t \dfrac{\pi(s)}{S(s)} \,\mathrm{d} S(s) - \dfrac{1}{2} \sum\limits_{i,j = 1}^{d} \int_0^t \dfrac{\pi_{i}(s) \pi_{j}(s)}{S_{i}(s)S_{j}(s)} \,\mathrm{d} [ S_{i}, S_{j} ] (s)\right).
 \end{equation}
 In our model-free version of portfolio theory, we do not wish to make assumptions on the structure of the covariations $[S_i, S_j]$ apart from their existence  \eqref{xin}. In particular, we do not assume that $[S_i, S_j](t)$ is absolutely continuous in $t$.
Growth rates and covariances, which in~\cite{FK08} can be taken as functions, therefore need to be modeled as measures.

 \begin{definition}
The \emph{covariance}  of the stocks in the market is described by the  positive semidefinite matrix-valued Radon measure $a = (a_{ij})_{1\leq i,j \leq d}$   
defined as
$$
 a_{ij}(\mathrm{d}t) := \mathrm{d} [ \log S_{i}, \log S_{j} ] (t)= \dfrac{1}{S_{i}(t) S_{j}(t)} \,\mathrm{d} [  S_{i},  S_{j} ] (t), 
 \quad i,j = 1,\ldots,d.
$$
  The \emph{excess growth rate} of a portfolio $\pi$ is defined as the signed Radon measure 
$$
  \gamma_{\pi}^{*}(\mathrm{d}t) := \dfrac{1}{2} \left( \sum\limits_{i=1}^{d} \pi_{i}(t) a_{ii}(\mathrm{d}t) -  \sum\limits_{i,j=1}^{d} \pi_{i}(t) \pi_{j}(t) a_{ij}(\mathrm{d}t) \right). 
$$
  For any portfolio $\pi$, 
we define the \emph{covariances}  of the individual stocks \emph{relative to the portfolio} $\pi$ as follows for $i,j = 1,\ldots,d$,
\begin{align}
\tau_{ij}^{\pi} (\mathrm{d}t) :&= (\pi(t) - e_{i})^\top a(\mathrm{d}t)(\pi(t) - e_{j}). \label{taupi1} 
 \end{align}
\end{definition}

 \begin{lemma}\label{Vpiloglemma} We have
$$
 \log V^{\pi}(t) =\int_0^t \pi(s) \,\mathrm{d}  \log S(s) + \gamma^{*}_{\pi}([0,t]) .$$
\end{lemma}

\begin{proof}  Using \eqref{log S Ito eq}, Lemma~\ref{log lemma}, the fact that $V^\pi(0)=1$, and the associativity of the Stieltjes and F\"ollmer integrals  from~\cite[Theorem I.6 b]{Widder} and \cite[Theorem 13]{Schied13}, respectively, we get
\begin{align}
  \log V^\pi(t) &= \int_0^t\pi(s)\ud  \log S(s) +\frac{1}{2}\sum_{i=1}^d\int_0^t \pi_i(s)\ud[\log S_i ](s)-\frac{1}{2}\sum_{i,j=1}^d\int_0^t\pi_i(s)\pi_j(s)\ud [\log S_i, \log S_j ](s),
\notag\end{align}
which implies the assertion via the definition of $\gamma^\ast_\pi$. \end{proof}

Applying the preceding lemma to the market portfolio yields the following results.

\begin{lemma}\label{market portfolio dynamics lemma}We have the following formulas for the dynamics of the market weights $\mu_i$.
\begin{enumerate}
\item $\displaystyle
\ud  \log \mu_i(t) = \big(e_i-\mu(t)\big)\ud  \log S(t)-\gamma^{*}_{\mu}(\mathrm{d}t) .
$
\item $\displaystyle
\tau_{ij}^\mu([0,t])=  [\log \mu_i, \log \mu_j ] (t)$ and $\ud[\mu_i,\mu_j](t)=\mu_i(t)\mu_j(t)\tau_{ij}^\mu(\ud t)$.
\item 
$\displaystyle
\ud \mu_i(t)=\mu_i(t)\big(e_i-\mu(t)\big)\ud \log S(t)-\mu_i(t)\gamma^{*}_{\mu}(\ud t) +\frac{1}{2}\mu_i(t) \tau_{ii}^\mu(\ud t)$. \end{enumerate}
\end{lemma}

\begin{proof}(a) The definition of $\mu_i$ and the formula for $V^\mu$ from Lemma~\ref{market portfolio lemma} yield that
\begin{equation}\label{log mui eq}
\begin{split}
\log \mu_i(t)&=\log S_i(t)-\log(S_1(t)+\cdots+S_d(t))\\
&=\log S_i(t)-\log V^\mu(t)-\log(S_1(0)+\cdots+S_d(0)).
\end{split}
\end{equation}
 Taking differentials and using Lemma~\ref{Vpiloglemma}  proves the claim.

(b) First, it follows, e.g., from \eqref{log mui eq} that $\log\mu\in QV^d$. Next, since $t\mapsto \gamma_\mu^*([0,t])$ is continuous and of bounded variation, it has vanishing quadratic variation. Hence,  (a) and  \cite[Remark 8 and Proposition 12]{Schied13} imply that 
\begin{align}
[\log \mu_i ] (t)
&=\Big[ \int_0^\cdot\big(e_i-\mu(s)\big)\ud \log S(s)\Big](t)\notag\\
&=\sum_{k,l=1}^d  \int_0^t\big((e_i)_k-\mu_k(s)\big)\big((e_i)_l-\mu_l(s)\big)\ud[\log S_k,\log S_l ] (t)\label{63}\\
&=\int_0^t\big(\mu(t)-e_i\big)^\top a(\ud t)\big(\mu(t)-e_i\big)=\tau_{ii}^\mu([0,t]).
\notag\end{align}
The polarization identity \eqref{polarization identity}  now yields the first claim in (b). The second one then follows with Lemma~\ref{log lemma}.

(c) Setting 
$$I(t):=\int_0^t\big(e_i-\mu(s)\big)\ud \log S(s)=\log S_i(t)-\log S_i(0)-\int_0^t\mu(s)\ud \log S(s)$$ and integrating (a) gives 
$
\mu_i(t)=\mu_i(0)\exp\left(I(t) -    \gamma^{*}_{\mu}([0,t])\right)$.
Using that $t\mapsto  \gamma^{*}_{\mu}([0,t])$ is of bounded variation and hence has vanishing quadratic variation, the pathwise It\^o formula  gives
$$
\mu_i(t)=\mu_i(0) + \int_0^t\mu_i(s)\ud I(s)+ \frac{1}{2}\int_0^t\mu_i(s)\ud [I](s)- \int_0^t\mu_i(s)\gamma^*_\mu(\ud s).$$
To deal with the first integral on the right-hand side, the associativity rule  of  \cite[Theorem 13]{Schied13} yields that 
$ \int_0^t\mu_i(s)\ud I(s)= \int_0^t\mu_i(s)\big(e_i-\mu(s)\big)\ud \log S(s)$. For the second integral, \eqref{63}  and  the associativity of the Stieltjes integral \cite[Theorem I.6 b]{Widder} imply that $\int_0^t\mu_i(s)\ud [I](s)=\int_0^t\mu_i(s)\tau_{ii}^\mu(\ud s)$. This yields (c).
\end{proof}

The proof  of the following lemma is left to the reader, since it   follows straightforwardly by adapting the  proof from \cite[Lemma 3.3]{FK08}.

\begin{lemma} \label{invarianceProp}
For any pair of portfolios $\pi$ and $\rho$ we have the following \emph{num{\'e}raire invariance property}  
\begin{equation}
\gamma^{*}_{\pi}(\mathrm{d}t) = \dfrac{1}{2} \left( \sum\limits_{i =1}^{d} \pi_{i}(t) \tau_{ii}^{\rho}(\mathrm{d} t) - \sum\limits_{i = 1}^{d} \sum\limits_{j =1}^{d} \pi_{i}(t) \pi_{j}(t) \tau^{\rho}_{ij}(\mathrm{d}t) \right).\notag
\end{equation}
\end{lemma}

The following lemma is valid if $\pi$ is a portfolio for $S$  and in addition an admissible integrand for $\log\mu$. Lemma~\ref{pi lemma} (applied with $\log\mu$ in place of $\log S$) states that the latter requirement is equivalent to $\frac\pi\mu$ being an admissible integrand for $\mu$. By Lemma~\ref{G generated portfolio lemma}, this requirement is satisfied for the functionally generated portfolio from \eqref{portfolioGeneratedByG}.

\begin{lemma}Suppose that $\pi$ is both a portfolio for $S$  and an admissible integrand for $\log\mu$. Then
\begin{align} \label{relReturn}
  \log\left(\dfrac{V^{\pi}(t)}{V^{\mu}(t)}\right) &= \int_0^t\dfrac{\pi(s)}{\mu(s)} \ud \mu(s) - \dfrac{1}{2}\sum_{i,j=1}^d\int_0^t \pi_{i}(s)\pi_{j}(s)\tau^{\mu}_{ij}(\mathrm{d}s).  
\end{align}
\end{lemma}

\begin{proof}
First, as noted above, $\frac\pi\mu$ is an admissible integrand for $\mu$.  Using Lemma~\ref{market portfolio dynamics lemma}  and the associativity of the Stieltjes integral  from~\cite[Theorem I.6 b]{Widder} together with the associativity of the Föllmer integral \cite[Theorem 13]{Schied13} yields that
\begin{align*}
 \frac{\pi(t)}{\mu(t)}\ud^{}\mu(t)&= \big(\pi(t)-\mu(t)\big)\ud^{}\log S(t)-\gamma_\mu^\ast(\ud t)+\frac{1}{2}\sum_{i=1}^d\pi_i(t)\tau_{ii}^\mu(\ud t),
\end{align*} 
due to the fact that the portfolio weights sum up to one.
Furthermore, applying the num$\acute{\text{e}}$raire invariance property  from Lemma~\ref{invarianceProp} gives us
\begin{align*} 
& \frac{\pi(t)}{\mu(t)}\ud  \mu(t)=\big(\pi(t)-\mu(t)\big)\ud  \log S(t)-\gamma_\mu^\ast(\ud t)+\dfrac{1}{2} \left( \sum_{i,j=1}^d\pi_{i}(t)\pi_{j}(t)\tau^{\mu}_{ij}(\mathrm{d}t)\right)
+\gamma_\pi^\ast(\ud t).
\end{align*}
On the other hand,  Lemma~\ref{Vpiloglemma} yields that 
 \begin{align*}
\ud \log\left(\dfrac{V^{\pi}(t)}{V^{\mu}(t)}\right) &= (\pi(t) - \mu(t)) \ud  \log S(t) + (\gamma^{*}_{\pi} - \gamma^{*}_{\mu})(\mathrm{d}t).
\end{align*}
Subtracting these formulas from each other now yields the assertion.
\end{proof}

As a preparation for the proof of Theorem~\ref{MasterFormulaCadlag}, the following lemma further calculates the  right-hand side of \eqref{relReturn} in case $\pi$ is  the functionally generated portfolio from \eqref{portfolioGeneratedByG}.

\begin{lemma}\label{aux int lemma}Let $G$ be as in Lemma~\ref{G generated portfolio lemma}, $\pi$ be the portfolio generated by $G$, and let us write 
\begin{equation}\label{nabla log G shorthand notation}
g(t)=(g_1(t),\dots,g_d(t))^\top:=\nabla_x\log G(\mu(t),A(t)).
\end{equation}
 Then
$$ \log\left(\dfrac{V^{\pi}(t)}{V^{\mu}(t)}\right)=\int_0^tg(s)\ud\mu(s)
-\frac12\sum_{i,j=1}^d\int_0^t\mu_i(s)\mu_j(s)g_i(s)g_j(s)\,\tau_{ij}^\mu(\ud s).$$
\end{lemma}

\begin{proof}First, we deal with the It\^o integral in \eqref{relReturn}. With the shorthand notation \eqref{nabla log G shorthand notation}, 
the definition \eqref{portfolioGeneratedByG} of $\pi$ becomes 
$\pi_i=\mu_i(t)(g_i(t)+1-\mu(t)^\top g(t))$, and so $\frac{\pi(t)}{\mu(t)}=g(t)+(1-\mu(t)^\top g(t)){\bm 1}$, where 
${\bm 1}:=(1,\dots,1)^\top\in\R^d$ denotes the vector whose entries are all $1$. Since both $\frac\pi\mu$ and $g$ are admissible integrands for $\mu$, the function $\frac\pi\mu-g=(1-\mu ^\top g){\bm 1}$ must also be an admissible integrand for $\mu$. But $\bm 1^\top(\mu(t)-\mu(s))=0$ for all $s$ and $t$, and therefore  the Riemann sums in the approximation \eqref{Foellmer integral Riemann sums} of the Föllmer integral $\int_0^t(1-\mu(s)^\top g(s)){\bm 1}\ud\mu(s)$ must all vanish. It follows that this  integral is zero, and hence that $ \int_0^t\frac{\pi(s)}{\mu(s)}  \mu(\mathrm{d}s)=\int_0^tg(s)\mu(\mathrm{d}s)$.

To deal with the rightmost integral in \eqref{relReturn}, we first note  that
 \begin{align*}
\tau_{ij}^{\mu} (\mathrm{d}t) &= a_{ij}(\mathrm{d}t) -\sum\limits_{\ell =1}^{d} \mu_{\ell}(t) a_{i\ell}(\mathrm{d}t) - \sum\limits_{k =1}^{d} \mu_{k}(t) a_{kj}(\mathrm{d}t) + \sum\limits_{k,\ell = 1}^{d} \mu_{k}(t) \mu_{\ell}(t) a_{k\ell}(\mathrm{d}t). \label{taupi2}\end{align*}
Thus, using the fact that the components of $\mu$ sum up to 1 we get 
  \begin{equation}\label{taupieq}
\sum\limits_{j =1 }^{d} \mu_{j}(t) \tau^{\mu}_{ij}(\mathrm{d}t) = 0, \quad i = 1,\ldots,d.
\end{equation}
Using that $\pi_i=\mu_i(t)(g_i(t)+1-\mu(t)^\top g(t))$, the preceding identity yields that
$$
\sum_{i,j=1}^d\pi_i(t)\pi_j(t)\tau^\mu_{ij}(\ud t)=\sum_{i,j=1}^dg_i(t)g_j(t)\mu_i(t)\mu_j(t)\tau^\mu_{ij}(\ud t).
$$
This proves the lemma.\end{proof}

\subsection{Proof of Theorem~\ref{MasterFormulaCadlag}}

Letting $g$ be as in \eqref{nabla log G shorthand notation} and using the fact that
$$\big(\log G)_{x_i,x_j}=\frac{G_{x_i,x_j}}{G}-g_ig_j
$$ 
the pathwise It\^ o formula and Lemma~\ref{market portfolio dynamics lemma} (b) yield
 \begin{align*}
\log \bigg(\frac{G(T,\mu(T),A(T))}{G(0,\mu(0),A(0)}\bigg)&= \int_0^Tg(t)\ud \mu(t)   + \sum_{k=0}^m\int_0^T\frac{\partial}{\partial{a_k}} \log G(t,\mu(t),A(t))  \ud A_k(t)\\&\qquad + \dfrac{1}{2} \sum_{i,j=1}^{d}\int_0^T\left( \dfrac{\partial_{ij}^{2}G(t,\mu(t),A(t))}{G(t,\mu(t),A(t))} - g_{i}(t) g_{j}(t) \right) \mu_{i}(t)\mu_{j}(t)\tau_{ij}^{\mu}(\mathrm{d}t).     \end{align*} 
 Comparing this formula with Lemma~\ref{aux int lemma} now gives the assertion.\hfill
\qedsymbol

\subsection{Proofs of the results from Section~\ref{path-dependent master formula section}
}

\begin{proof}[Proof of Lemma~\ref{functional portfolio value lemma}]
The chain rule \cite[Lemma 3.3]{Vol15} gives that  $\bbar V^\pi\in\mathbb C^{1,2}_c(\mathbb R^d_+,W\times\mathbb R)$ and that $\nabla_X \bbar V^\pi(t,X,(A,B^\pi))=\xi(t,X)$. Therefore, $\xi$ is indeed an admissible integrand on $\mathbb R^d_+$. 
Next, applying once again the chain rule \cite[Lemma 3.3]{Vol15} gives
\begin{align*}\lefteqn{\partial_{ij}\bbar V^\pi(t,X,(A,B^\pi))=\partial_i\xi_j(t,X)}\\
&=\frac{\bbar V^\pi(t,X,(A,B^\pi))}{X_i(t)X_j(t)}\Big((\partial_{ij}F)(t,\log X,A)+\pi_i(t,\log X)\pi_j(t,\log X)-\delta_{ij}\pi_i(t,\log X)\Big)
\end{align*}
and
\begin{align*}
\mathscr{D}_i\bbar V^\pi(t,X,(A,B^\pi))&=\bbar V^\pi(t,X,(A,B^\pi))\mathscr{D}_iF(t,\log X,A),\qquad i=0,\dots, m,\\
\mathscr{D}_{m+1}\bbar V^\pi(t,X,(A,B^\pi))&=-\bbar V^\pi(t,X,(A,B^\pi)).
\end{align*}
The functional It\^o formula, the formula \eqref{log covariation id}, and the associativity of the Stieltjes integral now  give  \eqref{funct portf sf cond}.
\end{proof}

\begin{proof}[Proof of Lemma~\ref{functional market portfolio lemma}] The functional  $H(t,X):=\log(e^{X_1(t)}+\cdots+e^{X_d(t)})$ clearly belongs to $\mathbb C^{1,2}_c(\mathbb R^d,\emptyset)$, and we have $\bbar\mu(t,X)=\nabla_XH(t,X)$. Therefore, $\bbar\mu$ is indeed an admissible functional integrand. It is also clear that $\bbar\mu_1+\cdots+\bbar\mu_d=1$. Next, we argue that 
\begin{equation}\label{market portfolio and f market portfolio id}
\int_0^t\bbar\mu(s,\log S)\ud\log S(s)=\int_0^t\mu(s)\ud\log S(s).
\end{equation}
That is, the functional It\^ o integral $\int_0^t\bbar\mu(s,\log S)\ud\log S(s)$ coincides with the F\"ollmer integral of the market portfolio $\mu(t)$ with respect to $\log S$. (This fact is not entirely obvious, because both integrals are defined as the respective limits of different \lq\lq Riemann sums''). 
To prove \eqref{market portfolio and f market portfolio id}, note that   $\mathscr{D}_0H(t,X)=0$ and 
\begin{equation}\label{H second derivatives}
\partial_{ij}H(t,X)=\partial_i\mu_j(t,X)=\delta_{ij}\bbar\mu_i(t,X)-\bbar\mu_i(t,X)\bbar\mu_j(t,X),
\end{equation}
where $\delta_{ij}$ is again the Kronecker delta. Hence, 
\begin{align*}
\lefteqn{\int_0^t\bbar\mu(s,\log S)\ud\log S(s)}\\
&=H(t,\log S)-H(0,\log S)-\frac12\sum_{i,j=1}^d\int_0^t\Big(\delta_{ij}\mu_i(s)- \mu_i(s) \mu_j(s)\Big)\ud [\log S_i,\log S_j](s)\\
&=\int_0^t \mu(s)\ud\log S(s),
\end{align*}
which gives \eqref{market portfolio and f market portfolio id}. The identity \eqref{market portfolio and f market portfolio id} and \cite[Proposition 12]{Schied13}  imply in turn that 
\begin{align*}
\bigg[\int_0\bbar\mu(s,\log S)\ud\log S(s)\bigg](t)&=\sum_{i,j=1}^d\int_0^t\mu_i(s)\mu_j(s)\ud [\log S_i,\log S_j](s)\\
&=\sum_{i,j=1}^d\int_0^t\bbar\mu_i(s,\log S)\bbar\mu_j(s,\log S)\ud [\log S_i,\log S_j](s),
\end{align*}
and this completes the proof that $\bbar\mu$ is a functional portfolio. Finally, it follows from $\mathscr{D}_0H(t,X)=0$  and \eqref{H second derivatives} that $B^{\bbar\mu}$ as defined in Lemma~\ref{functional portfolio value lemma} with $F:=H$ vanishes identically. Thus, that lemma yields the claimed formula for $\bbar V^{\bbar\mu}$.
\end{proof}

\begin{proof}[Proof of Lemma~\ref{pathdependent functionally generated portfolio}]
 Let $\Gamma (t,X,A):=\log G(t,\bbar\mu(\cdot,X),A)$. The chain rule \cite[Lemma 3.3]{Vol15} implies that $\Gamma \in\mathbb C^{1,2}_c(\mathbb R^d,W)$ and that its $i^{\text{th}}$ partial vertical derivative is given by
\begin{align}
\partial_i\Gamma (t,X,A)&=\sum_{j=1}^dg_j(t,\bbar\mu(\cdot,X))\partial_i\bbar\mu_j(t,X)\label{Gamma first vert der}\\
&=\bbar\mu_i(t,X)\Big(g_i(t,\bbar\mu(\cdot,X))-\sum_{j=1}^d\bbar\mu_j(t,X)g_j(t,\bbar\mu(\cdot,X))\Big)\nonumber\\
&=\pi_i(t,X)-\bbar\mu_i(t,X),\nonumber
\end{align}
where we have used \eqref{H second derivatives} in the second step.
Therefore, $\pi-\bbar\mu$ is an admissible functional integrand on $\mathbb R^d$. Finally, by Lemma~\ref{functional market portfolio lemma},  $\bbar\mu$, and in turn $\pi$, are admissible functional integrands on $\mathbb R^d$.
\end{proof}

We will also need the second vertical derivatives of the functional $\Gamma$ introduced in the proof of the preceding lemma.

\begin{lemma}\label{Gamma second derivatives lemma}For $G$ as in Lemma~\ref{pathdependent functionally generated portfolio} and $\Gamma (t,X,A)=\log G(t,\bbar\mu(\cdot,X),A)$, we have 
\begin{equation}\label{Gamma second derivatives eq}
\begin{split}
\partial_{ij}\Gamma &=\sum_{k,\ell=1}^d \frac{\partial_{\ell k}G}G\bbar\mu_\ell\bbar\mu_k\big(\bbar\mu_j-\delta_{j\ell}\big)\big(\bbar\mu_i-\delta_{ik}\big)-\pi_i\pi_j+\bbar\mu_i\bbar\mu_j+\delta_{ij}(\pi_i-\bbar\mu_i).
\end{split}
\end{equation}
\end{lemma}

\begin{proof}Using \eqref{Gamma first vert der} and the chain rule \cite[Lemma 3.3]{Vol15} , we find that
\begin{align*}
\partial_{ij}\Gamma (t,X,A)&=\partial_j\sum_{k=1}^dg_k(t,\bbar\mu(\cdot,X))\partial_i\bbar\mu_k(t,X)\\
&=\sum_{k=1}^d\bigg(\sum_{\ell=1}^d\partial_\ell g_k(t,\bbar\mu(\cdot,X))\partial_j\bbar\mu_\ell(t,X)\partial_i\bbar\mu_k(t,X)+g_k(t,\bbar\mu(\cdot,X))\partial_{ij}\bbar\mu_k(t,X)\bigg).
\end{align*}
By \eqref{H second derivatives}, we have  $\partial_j\bbar\mu_k=\delta_{jk}\bbar\mu_k-\bbar\mu_k\bbar\mu_j$ and hence
$$\partial_{ij}\bbar\mu_k=\mu_k\delta_{ik}\delta_{jk}-\mu_{k}\mu_i\delta_{k j}-\mu_j\mu_k\delta_{k i}-\mu_k\mu_i\delta_{ij}+2\mu_i\mu_j\mu_k.
$$
Therefore, when letting $\wt{\pi}=\pi-\bbar\mu$,
\begin{align*}
\sum_{k=1}^dg_k\partial_{ij}\bbar\mu_k=\Big(\delta_{ij}\mu_i-\bbar\mu_i\bbar\mu_j\Big)\Big(g_i-\sum_{k=1}^dg_k\bbar\mu_k\Big)-\bbar\mu_i\bbar\mu_j\Big(g_j-\sum_{k=1}^dg_k\bbar\mu_k\Big)=\wt\pi_i\delta_{ij}-\bbar\mu_j\wt\pi_i-\bbar\mu_i\wt\pi_j.
\end{align*}
 Moreover, 
$\partial_\ell g_k=\frac{\partial_{\ell k}G}G-g_\ell g_k$,
and so
\begin{align*}
\sum_{k,\ell=1}^d\partial_\ell g_k\partial_j\bbar\mu_\ell \partial_i\bbar\mu_k=\sum_{k,\ell=1}^d \frac{\partial_{\ell k}G}G\big(\delta_{j\ell}\bbar\mu_\ell-\bbar\mu_\ell\bbar\mu_j\big)\big(\delta_{ik}\bbar\mu_k-\bbar\mu_k\bbar\mu_i\big)-\bbar\mu_i\bbar\mu_j\Big(g_i-\sum_{k=1}^dg_k\bbar\mu_k\Big)\Big(g_j-\sum_{k=1}^dg_k\bbar\mu_k\Big).\end{align*}
Putting everything together yields the assertion after a short computation.
\end{proof}

\begin{proof}[Proof of Theorem~\ref{Path-dependent pathwise master formula}] 
As in Lemma~\ref{functional portfolio value lemma}, we consider the functionals $\bbar V^\pi$ and $\bbar V^\mu$. It was shown in Lemma~\ref{pathdependent functionally generated portfolio} that $B^{\bbar\mu}=0$ and 
$$
\bbar V^\mu(t,X)=\exp\big(H(t,\log X)-H(0,\log X)\big),
$$
where $H(t,X)=\log(e^{X_1(t)}+\cdots+e^{X_d(t)})$. Moreover, it was shown in the proof of Lemma~\ref{pathdependent functionally generated portfolio} that $\pi=\nabla_X(\Gamma+H)$, where $\Gamma$ is as in Lemma~\ref{Gamma second derivatives lemma}. Hence, by Lemma~\ref{functional portfolio value lemma},
$$\log \left( \dfrac{V^{\pi}(T)}{V^{\mu}(T)} \right) =\Gamma(T,\log S,A)-\Gamma(0,\log S,A)-B^\pi(T)=\log \left( \dfrac{G( T,\mu,A)}{G( 0,\mu,A)} \right)-B^\pi(T).
$$
It thus remains to compute $B^\pi(T)$. By \eqref{Bpi eq}, \eqref{H second derivatives}, and \eqref{Gamma second derivatives eq}, we have
\begin{align*}
\lefteqn{\ud B^\pi(t)+ \mathfrak{h}(\ud t)}\\
&=  \frac{1}{2}\sum_{i,j=1}^d\bigg( (\partial_{ij} \Gamma)(t, \log S  ,A)+(\partial_{ij} H)(t, \log S)\\
&\qquad\qquad\qquad+\pi_i(t,\log S)\pi_j(t,\log S)-\delta_{ij}\pi_i(t,\log S)\bigg) \ud [\log S_i,\log S_j](t)\\
&= \frac{1}{2}\sum_{i,j,k,\ell=1}^d \bigg(\frac{(\partial_{\ell k}G)(t, \log S  ,A)}{G(t, \log S  ,A)}\big(\mu_\ell(t)\mu_j(t)-\delta_{j\ell}\mu_\ell(t)\big)\big(\mu_k(t)\mu_i(t)-\delta_{ik}\mu_k(t)\big)\bigg)\ud [\log S_i,\log S_j](t).
\end{align*}
Recalling the notation $a_{ij}(\ud t)=\ud [\log S_i,\log S_j](t)$ and $\tau^\mu_{\ell k}(\ud t)=\sum_{i,j=1}^d(\mu_j-\delta_{j\ell })a_{ij}(\mu_i-\delta_{ik})$ as well as the fact that  $\ud [\mu_\ell,\mu_k](t)=\mu_\ell(t)\mu_k(t)\tau_{\ell k}(\ud t)$ by Lemma~\ref{market portfolio dynamics lemma}, we finally arrive at the desired identity $\ud B^\pi(t)+ \mathfrak{h}(\ud t)=-\mathfrak{g}(\ud t)$.
 \end{proof}

\subsection{Proof of Proposition~\ref{Examples eq prop}}
By the  chain rule for functional derivatives, \cite[Lemma 3.3]{Vol15}, $\log\wh G\in\mathbb C_c^{1,2}(U,\emptyset)$, and its vertical derivative  is given by
\begin{equation}
\partial_i\log\wh G(t,X)=\frac\lambda{\wh G(t,X)}
{ \varphi _{x_i}\Big(\lambda X(t)+\frac{1-\lambda}\theta\int_{t-\theta}^tX(0\vee s)\ud s\Big)}.
\end{equation}
When evaluating this expression at $X=\mu$, it becomes equal to $g_i(t)$. Since we know that $\mu(t)=\bbar\mu(t,\log S)$, the identity \eqref{examples general portfolio eq} follows. Later on in this proof, we will also need the horizontal derivative $\mathscr{D}_0\log\wh G(t,X)$. By \eqref{horizontal derivative 0}, it is given by
\begin{align*}
\mathscr{D}_0\log\wh G(t,X)=\frac{1-\lambda}{\theta \wh G(t,X)}\sum_{i=1}^d\varphi _{x_i}\Big(\lambda X(t)+\frac{1-\lambda}\theta\int_{t-\theta}^tX(0\vee s)\ud s\Big)\big(X_i(t-)-X_i(0\vee(t-\theta)-)\big),
\end{align*}
where we put $X(0-)=X(0)$. Therefore,
\begin{equation}\label{horizontal der log wh G}
\mathscr{D}_0\log\wh G(t,\mu)=\frac{1-\lambda}{ \lambda}\sum_{i=1}^dg_i(t)\alpha_i'(t),
\end{equation}

In analogy to the proof of Lemma~\ref{functional market portfolio lemma}, it suffices to establish the identity
\begin{equation}\label{functional ordinary pi Ito eq}
\int_0^t\wh\pi(s,\log S)\ud\log S(s)=\int_0^t\pi(s)\ud\log S(s)
\end{equation}
so as to conclude that $\wh\pi$ satisfies \eqref{path portolio quadr var eq} and is hence a functional portfolio. To this end, recall from the proof of Lemma~\ref{pathdependent functionally generated portfolio} that $\wh\pi(t,X)=\nabla_X\Gamma(t,X)+\nabla_XH(t,X)$, where $\Gamma(t,X)=\log\wh G(t,\bbar\mu(\cdot,X))$ and $H(t,X)=h(X(t))$ for $h(x)=\log (e^{x_1}+\cdots+e^{x_d})$. Since we already know from the proof of Lemma~\ref{functional market portfolio lemma} that $\int_0^t\nabla_XH(s,\log S)\ud\log S(s)=\int_0^t\mu(s)\ud\log S(s)$, it will be enough to show that 
\begin{equation}\label{Gamma gamma id}
\int_0^t\nabla_X\Gamma(s,\log S)\ud\log S(s)=\int_0^t\nabla_x\gamma(\log S(t),\alpha(t))\ud\log S(s),
\end{equation}
where $\gamma(x,a)=\log G(\nabla_xh(x),a)$ is as in the proof of Lemma~\ref{G generated portfolio lemma}. The functional It\^o formula implies that the left-hand side of \eqref{Gamma gamma id} is given by
\begin{align*}
\Gamma(t,\log S)-\Gamma(0,\log S)-\int_0^t\mathscr D_0\Gamma(s,\log S)\ud s-\frac12\sum_{i,j=1}^d\int_0^t\partial_{ij}\Gamma(s,\log S)\ud[\log S_i,\log S_j](s).
\end{align*}
We clearly have $\Gamma(t,\log S)=\gamma(\log S(t),\alpha(t))$. 
  Moreover, Lemma~\ref{Gamma second derivatives lemma} implies that $\partial_{ij}\Gamma(t,\log S)=\gamma_{x_i,x_j}(\log S(t),\alpha(t))$. Next, the chain rule for functional derivatives \cite[Lemma 3.3]{Vol15} yields that 
  \begin{align*}
  \mathscr D_0\Gamma(t,X)= (\mathscr D_0\log \wh G)(t,\bbar\mu(\cdot,X))+\sum_{i=1}^d(\partial_i\log \wh G)(t,\bbar\mu(\cdot,X))\mathscr D_0\bbar\mu_i(t,X)= (\mathscr D_0\log \wh G)(t,\bbar\mu(\cdot,X)),
  \end{align*}
since $\mathscr D_0\bbar\mu_i=0$. Together with \eqref{horizontal der log wh G} we thus obtain
\begin{equation}
 \mathscr D_0\Gamma(t,\log S)=\frac{1-\lambda}{ \lambda}\sum_{i=1}^dg_i(t)\alpha_i'(t).
\end{equation}
Re-writing the F\"ollmer integral on the right-hand side of \eqref{Gamma gamma id} by means of the It\^o formula for $\gamma(\log S(t),\alpha(t))$ and putting everything together thus yields \eqref{Gamma gamma id}.

Next, by the computations we have already completed in this proof, it is clear that
\begin{align*}
\mathfrak{h}([0,T])&=-\int_0^T\frac{G_a(\mu(t),\alpha(t))}{G(\mu(t),\alpha(t))}\ud \alpha(t)=-\frac{1-\lambda}{ \lambda}\sum_{i=1}^d\int_0^Tg_i(t)\alpha_i'(t)\ud t=-\int_0^T\frac{\mathscr D_0\wh G(t,\mu)}{G(t,\mu)}\ud t,\\
\mathfrak{g}([0,T])&= - \frac12\sum_{i,j=1}^{d}  \int_0^T\frac{G_{x_i,x_j}( \mu(t),\alpha(t))}{ G( \mu(t),\alpha(t))} \ud [\mu_i,\mu_j](t)= - \frac12\sum_{i,j=1}^{d}  \int_0^T\frac{\partial_{ij}\wh G( \mu(t),\alpha(t))}{ \wh G( \mu(t),\alpha(t))} \ud [\mu_i,\mu_j](t)\\
&= - \frac{\lambda^2}2\sum_{i,j=1}^{d}  \int_0^T\frac{\varphi_{x_i,x_j} (\lambda \mu(t)+(1-\lambda)\alpha(t))}{ \varphi (\lambda \mu(t)+(1-\lambda)\alpha(t))} \ud [\mu_i,\mu_j](t)
\end{align*}Since, moreover, $G(\mu(t),\alpha(t))=\wh G(t,\mu)$, Theorems~\ref{MasterFormulaCadlag} and~\ref{Path-dependent pathwise master formula} imply that $V^\pi(T)=V^{\wh\pi}(T)$.
\hfill\qedsymbol

\begin{rem}As mentioned before, a main difficulty in dealing with \emph{functional} pathwise It\^o calculus is the fact that  here  the pathwise It\^o integral is then no longer the limit of ordinary Riemann sums as in 
\eqref{Foellmer integral Riemann sums}. Instead, the integrands in the approximating \lq\lq Riemann sums'' \eqref{Itointcont} involve approximations of the integrator path. Ananova and Cont \cite[Theorem 3.2]{Ananova} provide strong regularity assumptions on both the integrand and the integrator to guarantee that  \eqref{Itointcont}  can be replaced by ordinary Riemann sums. In this context, it is interesting to note that  in the special cases \eqref{market portfolio and f market portfolio id} and \eqref{functional ordinary pi Ito eq} we could obtain similar results   without the regularity conditions on the integrators required in~\cite{Ananova}.
\end{rem}


\appendix  \label{A}
\section{Appendix on pathwise It\^o calculus}\label{Appendix} 

For the convenience of the reader, we     give here a short overview of the definitions and notations of  pathwise functional It\^o calculus as developed by Dupire~\cite{Dupire} and Cont and Fournier~\cite{CF,CF13}. Our presentation and notation is close to~\cite{Vol15}. 
In the sequel, we fix $T>0$ and  open sets $U\subset  \mathbb{R}^d$  and $V \subset \mathbb{R}^m$.  The Skorokhod  space $D([0,T],U)$ will be equipped with the supremum norm $\|X\|_\infty=\sup_{u\in[0,T]}\arrowvert X(u)\arrowvert$.
 For $X\in D([0,T],U)$ and $t\in [0,T]$, we let $X^t=(X(t\wedge s))_{s\in[0,T]}$  denote the  path stopped in $t$.  A functional $F:[0,T]\times D([0,T],U)\times   \text{\sl CBV}([0,T],V)  \mapsto \mathbb{R}$ is called \emph{non-anticipative} if 
$
 F(t, X, A)=F(t, X^t, A^t)
$
for all $(t,X, A)\in [0,T]\times D([0,T],U)\times   \text{\sl CBV}([0,T],V)$. We now recall several regularity properties for non-anticipative (and possibly vector-valued) functionals $F$. 
\begin{itemize}
\item $F$ 
is called \emph{boundedness-preserving}   if for every $A\in \text{\sl CBV}([0,T],V)$ and any compact subset $K\subset U$  there exist a constant $C$ such that 
$\arrowvert F(t,X,A)\arrowvert \le  C$ for all $t\in [0,T]$ and  $X\in D([0,T],K)$.
\item  $F$ is called \emph{continuous at fixed times}, if for all $\eps>0$, $t\in[0,T]$, $X\in  D([0,T],U)$, and $A\in \text{\sl CBV}([0,T],V)$, there exists $\eta>0$ such that $|F(t,X,A)-F(t,Y,A)|<\eps$ for all $Y\in D([0,T],U)$ for which $\|X^t-Y^t\|_\infty<\eta$.
\item  $F$ 
is called \emph{left-continuous}  if for all $t\in(0,T]$, $\eps>0$, $X\in D([0,T],U)$, and $A\in \text{\sl CBV}([0,T],V)$, there exists $\eta>0$ such that 
$
|F(t,X,A)-F(t-h,Y,A)|<\eps
$ for all $h\in[0,\eta)$ and $Y\in D([0,T],U)$ for which
$\|X^t-Y^{t-h}\|_\infty<\eta$. 
\item $F$  is called \emph{continuous in $X$  locally uniformly in $t$,}  if for all $\eps>0$ and $(t,X,A) \in [0,T]\times D([0,T],U)\times\text{\sl CBV}([0,T],S)$ there is some $\eta>0$ such that 
$|F(u,X,A)-F(u,Y,A)|<\eps$
for all $(u,Y) \in [0,T]\times D([0,T],U)$ for which
$\|X -Y \|_\infty<\eta$ and $  |t-u|<\eta$.
\end{itemize}

Next, we recall the notions of horizontal and vertical derivatives, which are also called Dupire derivatives and which were proposed in~\cite{Dupire,CF}. The following notion of a horizontal derivative extends  the one from~\cite{Dupire,CF} and was proposed in~\cite{Vol15}.  
We say that $F$ is \emph{horizontally  differentiable}, 
if there exist  non-anticipative and boundedness preserving functionals $\mathscr{D} _0F, \mathscr{D} _{1}F,\dots, \mathscr{D} _{m}F$ on $[0,T]\times  D([0,T],U)\times \text{\sl CBV}([0,T],V) $ such that for $0\le s<t\le T$ and 
$(X,A)\in  D([0,T],U)\times \text{\sl CBV}([0,T],V) $, the functions $[s,t]\ni r\mapsto \mathscr{D} _iF(r,X^s,A)$ are Borel measurable and
\begin{align}\label{horizontal derivatives}
 F(t,X^s,A)-F(s,X^s,A)=\sum_{i=0}^m\int_s^t\mathscr{D} _iF(r,X^s,A)\, A_i(dr),
\end{align}
where we put $A_0(r):=r$.
As discussed in \cite[Remark 2.2]{Vol15},  $F$ will be horizontally differentiable with horizontal derivative $\mathscr{D}F=\left(\mathscr{D} _0F, \mathscr{D} _{1}F,\dots, \mathscr{D} _{m}F\right) $, if the following limits exist for all $(t,X,A)$ and if they give rise to locally bounded and non-anticipative functionals  on $[0,T]\times  D([0,T],U)\times \text{\sl CBV}([0,T],V) $ satisfying the above measurability requirement,
\begin{align}
\mathscr{D} _0F(t,X^t,A^t)&=\lim_{h\downarrow 0}\frac{F(t+h,X^{t},A^{t})-F(t,X^{t},A^{t})}{h}\label{horizontal derivative 0}\\
 \mathscr{D} _{k}F(t, X^{t}, A^{t})&=\lim_{h\downarrow 0}\frac{F(t,X^{t},A_1^{t},\dots, A_k^{t+h},\dots,A_m^{t})-F(t,X^{t},A^{t})}{A_k(t+h)-A_k(t)} \Ind{\{A_k(t+h)\neq A_k(t)\}},\nonumber
 \end{align}
for $ k=1,\;\dots, m$. 

A non-anticipative functional $F$ 
is said to be \emph{vertically differentiable} at $(t,X,A) $ if the map
$
\mathbb{R}^d\ni v\to F(t, X+v\Ind{[t,T]}, A^t)
$
is differentiable at $0$. The \emph{vertical derivative} of $F$ at $(t,X,A) $ will then be the gradient of that map at $v=0$. It will be denoted by
\begin{equation}\label{vertical derivative}
\nabla_X F(t,X,A)=\left( \partial _i F(t,X,A)\right)_{i=1,\dots, d},
\end{equation}
where $\partial _i F(t,X,A)$ is the $i^{\text{th}}$ partial vertical derivative,
$$ \partial _i F(t,X,A)=\lim_{h\to0}\frac{F(t,X+he_i\Ind{[t,T]},A)-F(t,X,A)}{h}.
$$

If the functional $F$ admits  horizontal and vertical derivatives $ \mathscr{D} F$ and $ \nabla_X F$, we may iterate the corresponding operations so as to define higher order horizontal and vertical derivatives. 
We  denote by $\mathbb{C}^{1,2}_b(U,V)$ the set of all non-anticipative functionals $F
$ on $[0,T]\times D([0,T],U)\times\text{\sl CBV}([0,T],V)$ such that $F$ is left-continuous, horizontally differentiable, and twice vertically differentiable;  the horizontal derivative $\mathscr DF$ is  continuous at fixed times; the vertical derivatives $\nabla_XF$ and $\nabla^2_XF$ are left-continuous and boundedness-preserving.  With $\mathbb C^{1,2}_c(U,V)$ we denote the class of all functionals $F\in \mathbb C^{1,2}_b(U,V)$ that are continuous in $X$  locally uniformly in $t$ and boundedness preserving. Functionals in $\mathbb{C}^{1,2}_b(U,V)$ satisfy the following pathwise It\^o formula, which is taken from~\cite{Vol15} and which slightly extends the ones from~\cite{Dupire, CF}.

\begin{thm} \label{cvfcont} Let us fix a path $X\in C([0,T],U)$ with continuous quadratic variation, a path $A\in\text{\sl CBV}([0,T],V)$, and  a functional $F\in \mathbb{C}^{1,2}_b(U,V)$.
For $n\in\mathbb N$, define the approximating path $X^n\in D([0,T],U)$ by
 \begin{equation} 
X^n(t):=\sum_{s\in \mathbb{T}_n}X(s')\mathbbm{1}_{[s,s')}(t)+X(T)\Ind{\{T\}}(t),\qquad 0\leq t\leq T,
\label{Xn1}\end{equation}
and let $X^{n, s-}:=\lim_{r\uparrow s}X^{n, r}$. Then  the \emph{pathwise It\^o integral} along $\left(\mathbb{T}_n\right)$, \begin{equation} 
\int_0^T \nabla_XF(s, X, A)\ud X(s):=\lim_{n\uparrow\infty} \sum_{s\in \mathbb{T}_n}\nabla_XF(s,X^{n,s-}, A)\cdot\left(X(s')-X(s)\right),
\label{Itointcont}\end{equation}
 exists and, with $A_0(t)=t$,
 \begin{equation}\label{Ito formula}
\begin{split} 
F(T,X ,A )-F(0,X ,A )&=\int_0^T \nabla_XF(s,X , A )\ud X(s)+\sum_{i=0}^m\int_0^T\mathscr{D}_iF(s,X,A)\ud A_i( s)\\
&\qquad+ \frac{1}{2}\sum_{i,j=1}^d\int_0^T\partial_{ij} F(s, X  ,A)\ud [X_i,X_j](s).
\end{split}
 \end{equation}
\end{thm}

 \parskip-0.5em\renewcommand{\baselinestretch}{0.9}
\end{document}